\documentclass[lettersize,journal]{IEEEtran}
\usepackage{amsmath,amsfonts}
\usepackage{amssymb}
\usepackage{algorithmic}
\usepackage{algorithm}
\usepackage{array}
\usepackage{textcomp}
\usepackage{stfloats}
\usepackage{url}
\usepackage{verbatim}
\usepackage{graphicx}
\usepackage{subfigure} 
\usepackage{color,xcolor}
\usepackage{cite}
\usepackage{stmaryrd,txfonts}
\usepackage{blkarray}
\usepackage{threeparttable}

\newtheorem{theorem}{Theorem}
\newtheorem{lemma}[theorem]{Lemma}

\newtheorem{definition}{Definition}

\newenvironment{proof}{{\qquad\noindent\it Proof:}\;}{\hfill $\square$\par}

\begin{document}

\title{Density Evolution Analysis of Generalized Low-density Parity-check Codes under a Posteriori Probability Decoder}

\author{Dongxu Chang, Zhiming Ma, Qingqing Peng, Guanghui Wang, Guiying Yan, Dawei Yin
\thanks{This work is partially supported by the National Key R\&D Program of China, (2023YFA1009600).

D. Chang, Q. Peng, G. Wang, and D. Yin are with the School of Mathematics, Shandong University, Shandong, China (e-mail: dongxuchang@mail.sdu.edu.cn; pqing@mail.sdu.edu.cn; ghwang@sdu.edu.cn; daweiyin2@outlook.com).

Z. Ma and G. Yan are with the Academy of Mathematics and Systems Science, CAS, University of Chinese Academy of Sciences, Beijing, 100190 China (e-mail: mazm@amt.ac.cn; yangy@amss.ac.cn).}
}



\maketitle

\begin{abstract}
In this study, the performance of generalized low-density parity-check (GLDPC) codes under the a posteriori probability (APP) decoder is analyzed. We explore the concentration, symmetry, and monotonicity properties of GLDPC codes under the APP decoder, extending the applicability of density evolution to GLDPC codes. On the binary memoryless symmetric channels, using the BEC and BI-AWGN channels as two examples, we demonstrate that with an appropriate proportion of generalized constraint (GC) nodes, GLDPC codes can reduce the original gap to capacity compared to their original LDPC counterparts. Additionally, on the BI-AWGN channel, we apply and improve the Gaussian approximation algorithm in the density evolution of GLDPC codes. By adopting Gaussian mixture distributions to approximate the message distributions from variable nodes and Gaussian distributions for those from constraint nodes, the precision of the channel parameter threshold can be significantly enhanced while maintaining a low computational complexity similar to that of Gaussian approximations. Furthermore, we identify a class of subcodes that can greatly simplify the performance analysis and practical decoding of GLDPC codes, which we refer to as message-invariant subcodes. Using the aforementioned techniques, our simulation experiments provide empirical evidence that GLDPC codes, when decoded with the APP decoder and equipped with the right fraction of GC nodes, can achieve substantial performance improvements compared to low-density parity-check (LDPC) codes.
\end{abstract}


\section{Introduction}
\IEEEPARstart{A}{s} a variant of the low-density parity-check (LDPC) codes introduced by Gallager \cite{gallager1962low}, generalized low-density parity-check (GLDPC) codes were first proposed by Tanner \cite{tanner1981recursive}. While retaining the sparse graph representation, GLDPC codes replace single parity-check (SPC) nodes with error-correcting block subcodes, known as generalized constraint (GC) nodes. The degrees of GC nodes match the length of their associated subcodes. Various linear block codes, including Hamming codes \cite{lentmaier1999generalized} \cite{hirst2002application}, Bose-Chaudhuri-Hocquengham (BCH) codes \cite{boutros1999generalized}, Hadamard codes \cite{yue2007generalized}, and others, can serve as subcodes in constructing GC nodes. GLDPC codes offer advantages such as the ability to employ more powerful decoders at GC nodes during decoding, resulting in improved performance \cite{liu2019probabilistic}, faster convergence \cite{mulholland2015design}, and reduced error floor \cite{liva2008quasi} \cite{mitchell2013minimum}.

In \cite{liu2019probabilistic}, the tradeoff between code rate and asymptotic performance of a class of GLDPC code ensembles with fixed degree distributions constructed by including a certain fraction of GC nodes in the graph is analyzed over the binary erasure channel (BEC). 
Through the analysis of the probabilistic peeling decoder (P-PD) algorithm, Liu \cite{liu2019probabilistic} accurately predicts the performance of GLDPC codes under maximum likelihood (ML) decoding on GC nodes. The analysis in \cite{liu2019probabilistic} demonstrates that when the proportion of GC nodes is appropriate, GLDPC codes can reduce the gap to capacity compared to their original LDPC counterparts. However, for a wider range of channels, such as binary input additive white Gaussian noise (BI-AWGN) channels, there is still a lack of sufficient analysis on the performance of GLDPC codes.

In LDPC codes, density evolution \cite{richardson2001capacity} serves as a powerful tool for analyzing the decoding performance. By showing the properties of centrality, symmetry, and monotonicity for LDPC codes under belief propagation (BP) decoding, the performance analysis of LDPC codes with block length tends to infinity can be greatly simplified. It was demonstrated that the decoding performance of LDPC codes with block length tends to infinity can be analyzed using the all-zero codeword on cycle-free graphs. Additionally, a threshold of the channel parameter can be determined, which is the minimum channel quality that supports reliable iterative decoding of asymptotically large codes drawn from the given code ensemble. 

In this paper, we extend the performance analysis of GLDPC codes to binary memoryless symmetric (BMS) channels. In pursuit of the optimal design of GLPDC codes with high-performance characteristics, this paper introduces a systematic methodology. Specifically, for GLDPC codes under the a posteriori probability (APP) decoder, we adopt a general approach analogous to density evolution in LDPC codes, and the concentration, symmetry, and monotonicity properties are demonstrated on GLDPC codes.

Based on the aforementioned properties, density evolution analyses are provided for GLDPC codes on BMS channels, and we show the analysis over the BEC and BI-AWGN channels as two examples. However, it's worth noting that due to the computational complexity associated with APP decoding, both practical decoding and theoretical analysis face considerable challenges. In this paper, we identify a class of error-correcting block codes in which the propagation of messages from each connected variable node follows a uniform and consistent pattern when they are used as subcodes on GC nodes. Consequently, the performance analysis and practical decoding of GLDPC codes can be significantly simplified. We refer to such subcodes as message-invariant subcodes. Remarkably, it can be established that several codes, including Hamming codes, Reed-Muller codes, extended BCH codes, and others, can all be classified as message-invariant subcodes.


Considering the complexity of computing the message distribution of GC nodes on the BI-AWGN channel, similar to the Gaussian approximation used for LDPC codes \cite{chung2001analysis}, we propose a technique for efficiently computing the channel parameter threshold. We improve the traditional Gaussian approximation approach by approximating the distribution of messages sent by variable nodes using Gaussian mixture distributions. This can significantly reduce the errors introduced by Gaussian approximation while maintaining low computational complexity. It is worth noting that this Gaussian mixture approximation approach can also be applied to LDPC codes, yielding favorable outcomes. 

By selecting a (6, 3)-linear block code and a (7, 4)-Hamming code from the message-invariant subcodes on the GC nodes, which were also used in \cite{liu2019probabilistic} to illustrate that a suitable proportion of GC nodes can reduce the gap to capacity compared to the base LDPC code, 
our density evolution results reveal consistency with the findings in \cite{liu2019probabilistic} for both the BEC and BI-AWGN channel. When the proportion of GC nodes is appropriate, GLDPC codes, under the APP decoder, exhibit a reduced gap to capacity compared to the base LDPC codes.

Finally, we compare GLDPC codes with an appropriate GC proportion to LDPC codes with the same code rate through simulation experiments. The results indicate a significant performance improvement of GLDPC codes under the APP decoder compared to LDPC codes with the same design rate.

The rest of the paper is organized as follows. In Section \uppercase\expandafter{\romannumeral2}, we introduce GLDPC code ensemble, the message-passing decoder, and dentisy evolution on LDPC codes. Section \uppercase\expandafter{\romannumeral3} presents density evolution on GLDPC codes, where the concentration, symmetry, and monotonicity properties are demonstrated. In Section \uppercase\expandafter{\romannumeral4}, we provide a density analysis for GLDPC codes over the BMS channels. The paper concludes in Section \uppercase\expandafter{\romannumeral5}.

\section{Preliminaries}
In this section, we introduce the GLDPC code ensembles that will be analyzed. Next, we provide a brief overview of the message-passing decoder and density evolution on LDPC codes.

\subsection{GLDPC Code Ensembles}

In \cite{tanner1981recursive}, Tanner extended the concept of LDPC codes by introducing the use of block codes as constraint nodes, which are referred to as GC nodes. Instead of connecting to an SPC constraint, the variable nodes connected to the GC nodes must satisfy the constraints within the subcode associated with the GC node. Similar to the work in \cite{liu2019probabilistic}, our analysis focuses on GLDPC codes in which a proportion $t$ of the constraint nodes are designated as GC nodes, while the remaining proportion of $1-t$ nodes maintain the role of SPC nodes. 

Denote the subcode of GLDPC ensembles as $\mathcal{C}$, where $\mathcal{C}$ is a linear block code with code length $K$. The $(\mathcal{C}, J, K, t)$ GLDPC ensemble is defined as follows:
\begin{definition}[$(\mathcal{C}, J, K, t)$ GLDPC ensemble]
Each element of the $(\mathcal{C}, J, K, t)$ GLDPC ensemble is defined by a Tanner graph, as depicted in Figure 1. Within this graph, there exist $n$ variable nodes with degree $J$ and $m=\frac{nJ}{K}$ constraint nodes with degree $K$, as $n$ tends to infinity. Among the constraint nodes, $tm$ are designated as GC nodes, with $\mathcal{C}$ as their subcodes, while $(1-t)m$ are SPC nodes. The creation of a random element in this ensemble involves the use of a uniform random permutation $\pi$ for the $nJ$ edges connecting the variable nodes to the constraint nodes. Specifically, 
following a similar approach to that in \cite{richardson2001capacity}, each node is assigned either $J$ or $K$ ``sockets'' based on whether it functions as a variable node or a constraint node. These variable and constraint sockets are distinctly labeled. The edges within the corresponding bipartite graph are denoted by pairs ${i,\pi(i)}$, where $i=1,\cdots,nJ$. Here, $i$ and $\pi(i)$ respectively represent the corresponding variable node socket and constraint node socket. Furthermore, label the variable nodes contained in the subcodes of length $K$ from 1 to $K$. The variable node connected to the $i_G$-th socket of a specific GC node $G$ corresponds to the $i_G$-th variable node in the corresponding subcode $\mathcal{C}$, $1\leq i_G \leq K$.

Denote $m'$ as the number of rows in the parity-check matrix of $\mathcal{C}$, the design rate $R$ of the $(\mathcal{C}, J, K, t)$ ensemble can be expressed as follows:
\begin{equation}
    R = \frac{n-(\frac{nJ}{K}tm'+\frac{nJ}{K}(1-t))}{n} = 1-\frac{J}{K}-t\frac{J}{K}(m'-1).
\end{equation}
For the $(\mathcal{C}, J, K, 0)$ GLDPC ensemble, which is an LDPC ensemble without GC nodes, we refer to it as the base LDPC code. Denote its design rate as $R_0$, which can be represented as:
\begin{equation}
    R_0 = 1-\frac{J}{K}.
\end{equation}
Therefore, as pointed out in \cite{liu2019probabilistic},
\begin{equation}
    R = R_0 - t(1-R_0)(m'-1).
\end{equation}
\end{definition}

\begin{figure}[!t]
\centering
\includegraphics[width=2.5in]{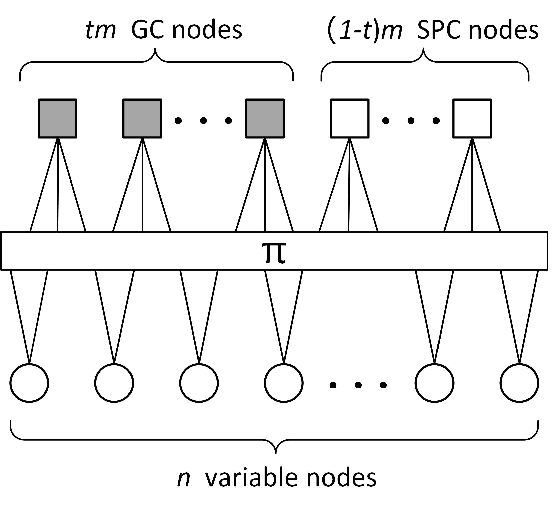}
\caption{The Tanner graph of GLDPC code ensembles.}
\label{fig_1}
\end{figure}

\subsection{Message-Passing Decoder}
Many successful decoding algorithms, such as the BP decoder and successive cancellation (SC) decoder \cite{arikan2009channel}, fall under the category of message-passing decoders. These decoders accomplish decoding through iterative message-passing processes. In the initial iteration, every variable node receives initial messages. Subsequently, in each iteration, each variable node processes all the messages it received to generate new messages, which it then transmits to its adjacent constraint nodes. Simultaneously, each constraint node processes the messages it receives and relays new messages back to each neighboring variable node.

Assume that the codeword $\boldsymbol{c}=\left\{c_1,\cdots,c_n\right\}\in \left\{0,1\right\}^n$ are transmitted through a BMS channel $W$, and $\boldsymbol{y}=\left\{y_1,\cdots,y_n\right\}$ are received signals. Let
	\begin{equation}
		L(y_i)=\ln{\frac{p(y_i|c_i=0)}{p(y_i|c_i=1)}}
	\end{equation}
denote the corresponding LLR of $y_i$.

For the BP decoder, the message update rule of 
variable-to-check (V2C) message is
	\begin{equation}
		L_{i \rightarrow \alpha}^{(l)}=L_{0}+\sum_{h \in N(i) \backslash \alpha} L_{h \rightarrow i}^{(l-1)},
  \label{V2C message update}
	\end{equation}
and the message update rule of check-to-variable (C2V) message is 
	\begin{equation}
		L_{\alpha \rightarrow i}^{(l)}=2 \tanh ^{-1}\left(\prod_{j \in N(\alpha) \backslash i} \tanh \left(L_{j \rightarrow \alpha}^{(l-1)} / 2\right)\right)
  \label{C2V message update}
	\end{equation}
where $L_0$ is the channel message in LLR form, $l$ is the number of iterations, $N(v)$ represents the nodes connected directly to node $v$, $i \rightarrow \alpha$ means from variable node $i$ to check node $\alpha$ and $\alpha \rightarrow i$ means from check node $\alpha$ to variable node $i$. The initial message $L_{i \rightarrow \alpha}^{(0)}$ and $L_{\alpha \rightarrow i}^{(0)}$ is 0.
 
In GLDPC code ensembles, the message-passing process of variable nodes in GLDPC codes is consistent with the case in LDPC codes, as shown in (\ref{V2C message update}). However, for the constraint nodes of GLDPC codes, they contain both SPC nodes and GC nodes, and these two types of nodes handle messages differently. For SPC nodes, the information they transmit is determined based on the information they receive and the SPC codes they correspond to, as described in (\ref{C2V message update}). In contrast, at the GC nodes, the transmission of messages is managed by their corresponding subcodes.

As a commonly used decoding algorithm on GC nodes \cite{lentmaier1999generalized} \cite{boutros1999generalized} \cite{yue2007generalized} with many simplified computation techniques \cite{bahl1974optimal} and variants \cite{mceliece1996bcjr} \cite{johansson1998simple}, the APP decoder holds a significant position in the study of GLDPC codes. Its decoding rules are as follows: For a GC node $G$ with subcode $\mathcal{C}$, let $v_i$ denote its $i$th neighboring variable node, $\boldsymbol{L}_{\thicksim i}$ denote the message received its neighboring variable nodes other than $v_i$, and $\boldsymbol{c'}$ denote the codewords of $\mathcal{C}$. The message to $v_i$ from $G$ under the APP decoder is
\begin{equation}
\begin{split}
    L_i &= \log\frac{P(\boldsymbol{L}_{\thicksim i}|v_i=0)}{P(\boldsymbol{L}_{\thicksim i}|v_i=1)}\\
    &=\log\frac{\sum\limits_{\boldsymbol{c'}:c'_i=0}P(\boldsymbol{L}_{\thicksim i}|\boldsymbol{c'})}{\sum\limits_{\boldsymbol{c'}:c'_i=1}P(\boldsymbol{L}_{\thicksim i}|\boldsymbol{c'})}\\
    &=\log\Bigg(\sum\limits_{\boldsymbol{c'}:c'_i=0} \exp\Big(\sum\limits_{j\neq i} \mathbb{I}[c'_j=0]\boldsymbol{L}_j\Big)\Bigg)\\
    &\;\:\;-\log\Bigg(\sum\limits_{\boldsymbol{c'}:c'_i=1} \exp\Big(\sum\limits_{j\neq i} \mathbb{I}[c'_j=0]\boldsymbol{L}_j\Big)\Bigg),
\end{split}
\label{APP}
\end{equation}
where $\mathbb{I}[c'_j=0]=1$ if $c'_j=0$, otherwise $\mathbb{I}[c'_j=0]=0$.

\subsection{Density Evolution}
Density evolution, as introduced by Richardson and Urbanke in \cite{richardson2001capacity}, can effectively characterize the asymptotic performance of the message-passing decoder, i.e., the performance when the code length $n$ tends to infinity. 

In the following, the notations used in \cite{richardson2008modern} will be employed.  Let $D_i$ be the density of $L(y_i)$ conditioned on $x_i=0$. We call $D_i$ the L-density of $y_i$. Let $l$ denote the number of decoding iterations.

Denote the convolution operations on the variable node and SPC node by two binary operators $\circledast$ and $\boxast$, respectively. For L-density $D_1$, $D_2$, and any Borel set $E\in \overline{\mathbb{R}}$, define
\begin{equation}
		(D_1 \circledast D_2)(E)\triangleq \int_{\overline{\mathbb{R}}} D_1(E-\alpha)D_2(d\alpha),
	\end{equation}
	\begin{equation}
		(D_1 \boxast D_2)(E)\triangleq \int_{\overline{\mathbb{R}}} D_1\left(2\tanh^{-1}\left(\frac{\tanh{\frac{E}{2}}}{\tanh{\frac{\alpha}{2}}}\right)\right)D_2(d\alpha),
	\end{equation}
where $\int_{\overline{\mathbb{R}}}f(\alpha)D(d\alpha)$ is the Lebesgue integral with respect to probability measure $D$ on extended real numbers $\overline{\mathbb{R}}$. 

If the factor graph is cycle-free, then the update rule in (\ref{V2C message update}), (\ref{C2V message update}) can be written in the form of L-density
\begin{equation}
    D_{(i,\alpha)}^{(l)} = D_i\circledast\Big(\circledast_{h\in N(i)\backslash \alpha} D_{(h,i)}^{(l-1)}\Big),
\end{equation}
\begin{equation}
    D_{(i,\alpha)}^{(l)} = \boxast_{j\in N(\alpha)\backslash i} D_{(j,\alpha)}^{(l-1)}.
\end{equation}

However, for continuous-output channels such as the BI-AWGN channel, the calculation of the message distribution becomes computationally challenging. In \cite{chung2001analysis}, Gaussian approximation is employed to simplify the analysis of the decoding performance. By approximating the distribution of messages as Gaussian distribution, without much sacrifice in accuracy, a one-dimensional quantity, namely, the mean of a Gaussian distribution, can act as a faithful surrogate for the message density, which in contrast, is an infinite-dimensional vector. 

Denote the means of the messages sent from variable nodes and check nodes in $l$-th iteration by $m_V^{(l)}$ and $m_S^{(l)}$, respectively. The variable nodes are with degree $d_v$ and the check nodes are with degree $d_c$. Then (\ref{V2C message update}) simply becomes 
\begin{equation}
    m_V^{(l)} = m_V^{(0)} + (d_v -1)m_S^{(l-1)},
\end{equation}
where $m_V^{(0)}$ is the mean of the channel message. The update rule for $m_S^{(l)}$ is
\begin{equation}
   m_S^{(l)}= \phi_S(m_V^{(l)}) = \phi^{-1}\Bigg(1-\Big[1-\phi(m_V^{(l)})\Big]^{d_c-1}\Bigg),
    \label{Gaussian_old_equation}
\end{equation}
where
\begin{equation}
\label{eq_ga1}
    \phi(x) =\left\{
\begin{aligned}
&   1-\frac{1}{\sqrt{4\pi x}}\int_{\mathbb{R}}\tanh \frac{u}{2}e^{-\frac{(u-x)^2}{4x}}du,\;\;\; if\;  x > 0,\\
&   1, \quad\quad\quad\quad\quad\quad\quad\quad\quad\quad\quad\quad\; if\; x=0.
\end{aligned}
\right.
\end{equation}

\section{Properties of Density Evolution on GLDPC Codes}
The density evolution algorithm introduced in \cite{richardson2001capacity} serves as an efficient tool for the asymptotic analysis of LDPC codes. In their seminal work \cite{richardson2001capacity}, Richardson and Urbanke demonstrated that the behavior of individual instances concentrates around its average behavior (of the code and the noise). This average behavior progressively aligns with the cycle-free LDPC code scenario as the code length increases. The validity of this centrality in LDPC codes is essentially determined by the nature of LDPC code BP decoding performed ``locally''. Since the nature of decoding locally also holds for ensembles of GLDPC codes, the concentration property remains applicable to GLDPC code ensembles. Consequently, the asymptotic performance analysis of GLDPC codes can be conducted on cycle-free graphs, and the incoming information to a node in GLDPC codes can be reasonably considered as mutually independent.



In this section, We theoretically justify the rationality of density evolution on GLDPC codes. We introduce and establish the symmetry conditions extended for GLDPC codes. We demonstrate that, under the extended symmetry conditions, the error probability of GLDPC codes becomes independent of the specific transmitted codeword. Moreover, we show that GLDPC codes, when decoded using the APP decoder, exhibit the desirable property of monotonicity, which ensures the existence of a channel threshold. 

\subsection{Symmetry Condition}
The symmetry conditions for LDPC codes outlined in \cite{richardson2001capacity} encompass various aspects, including channel symmetry, check node symmetry, and variable node symmetry. However, in the context of GLDPC codes, where check nodes are categorized into SPC nodes and GC nodes, it becomes necessary to extend the check node symmetry condition to suit the specific characteristics of GLDPC codes.

For simplicity, assume the code sequence $\boldsymbol{c} = $$\{c_1,c_2,\cdots,c_n\}$ is transmitted by BPSK modulation,  which maps the codeword into a bipolar sequence $\boldsymbol{x} = \{x_1,x_2,\cdots,x_n\}$, according to $x_i = 1-2c_i,i\in[1,n]$. Without causing confusion, we also refer to $\boldsymbol{x}$ as the codeword in the rest of this paper. To be concrete, consider a discrete case and assume that the output alphabet is 
\begin{equation}
    \mathcal{O} \triangleq \{-q_o, -(q_o-1),\cdots,-1,0,1,\cdots,(q_o-1),q_o\}
\end{equation}
and that the message alphabet is 
\begin{equation}
    \mathcal{M} \triangleq \{-q,-(q-1),\cdots,-1,0,1,\cdots,(q-1),q\}.
\end{equation}

Let $\Psi_V^{(l)}:\mathcal{O}\times\mathcal{M}^{J-1}\rightarrow\mathcal{M}$, $l\geq 1$, denote the variable node message map, let $\Psi_S^{(l)}:\mathcal{M}^{K-1}\rightarrow\mathcal{M}$, $l\geq 0$, denote the SPC node message map, and let $\Psi_{Gi}^{(l)}:\mathcal{M}^{K-1}\rightarrow\mathcal{M}$, $l\geq 0$, denote the GC node message map to its $i$-th connected variable node, $i\in\{1,\cdots,K\}$, as a function of $l\in\mathbb{N}$. For completeness, let $\Psi_V^{(0)}:\mathcal{O}\rightarrow\mathcal{M}$ denote the initial message map.

\quad\par

\begin{definition}[Extended Symmetry Conditions]
\label{def_Extended Symmetry Conditions}
\begin{itemize}
\item{\textbf{Channel symmetry:}} The channel is output-symmetric, i.e.,
\begin{equation}
    p(y_T=q|x_T = 1) = p(y_T=-q|x_T = -1),
\end{equation}
for all $T \in [1,n]$.

\item{\textbf{Variable node symmetry:}} Signs inversion invariance of variable node message maps holds
\begin{equation}
    \Psi_V^{(l)}(-m_0,-m_1,\cdots,-m_{J-1}) = -\Psi_V^{(l)}(m_0,m_1,\cdots,m_{J-1}),
\end{equation}
$l\geq1$, and
\begin{equation}
    \Psi_V^{(0)}(-m_0) = -\Psi_V^{(0)}(m_0).
\end{equation}
{\color{black}
where $m_T,T\in [1,J-1]$ denotes the message received from its $J-1$ neighboring check nodes.
} 
\item{\textbf{SPC node symmetry:}} Signs factor out of SPC node message maps
\begin{equation}
    \Psi_S^{(l)}(b_1m_1,\cdots,b_{K-1}m_{K-1}) = \Psi_S^{(l)}(m_1,\cdots,m_{K-1})\Bigg(\prod \limits_{i=1}^{K-1}b_i\Bigg)
\end{equation}
for any $\pm1$ sequence $(b_1,\cdots,b_{K-1})$ and {\color{black}$m_T,T\in [1,K-1]$ denotes the message received from its $K-1$ neighboring variable nodes.
}
\item{\textbf{GC node symmetry:}} The sign of the variable to receive information factors out of GC node message maps
\begin{equation}
\begin{split}
    &\Psi_{Gi}^{(l)}(b_1m_1,\cdots,b_{i-1}m_{i-1},b_{i+1}m_{i+1},\cdots,b_{K}m_{K})\\
    = &b_i\cdot\Psi_{Gi}^{(l)}(m_1,\cdots,m_{i-1},m_{i+1},\cdots,m_{K}),
\end{split}
\end{equation}
if the sequence $(b_1,\cdots,b_{K})$ is a codeword of the corresponding subcode $\mathcal{C}$.
\end{itemize}
\end{definition}

We follow the patterns in \cite{richardson2001capacity} and \cite{richardson2008modern} to prove the following lemmas.

\quad\par

\begin{lemma}[Conditional Independence of Error Probability Under Symmetry]
    Let $\mathcal{G}$ be the Tanner graph of a given GLDPC code and for a given message-passing algorithm let $P_e^{(l)}(\boldsymbol{x})$ denote the conditional (bit or block) probability of error after the $l$-th decoding iteration, assuming that codeword $\boldsymbol{x}$ was sent. If the channel and the decoder fulfill the extended symmetry conditions stated in Definition \ref{def_Extended Symmetry Conditions} then $P_e^{(l)}(\boldsymbol{x})$ is independent of $\boldsymbol{x}$. 
\end{lemma}
\begin{proof}
    Let $p(q)$ denote the channel transition probability $p(y=q|x = 1)$. Then any binary-input memoryless output-symmetric channel can be modeled multiplicatively as 
    \begin{equation}
        y_T = x_Tz_T,
    \end{equation}
    where $x_T$ is the input bit, $y_T$ is the channel output, and $z_T$ are i.i.d. random variables with distribution defined by $Pr\{z_T=q\}=p(q)$. Let $\boldsymbol{x}$ denote an arbitrary codeword and let $\boldsymbol{y}=\boldsymbol{x}\boldsymbol{z}$ be an observation from the channel after transmitting $\boldsymbol{x}$, where $\boldsymbol{z}$ denotes the channel realization (multiplication is componentwise and all three quantities are vectors of length $n$).

    Let $v_i$ denote an arbitrary variable node. Let $s_j$ denote one of its neighboring SPC nodes and let $g_k$ denote one of its neighboring GC nodes. For any received sequence $\boldsymbol{w}$, let $m_{ij}^{(l)}(\boldsymbol{w})$ denote the message sent from $v_i$ to $s_j$ in iteration $l$ assuming $\boldsymbol{w}$ was received and let $m_{ik}^{(l)}(\boldsymbol{w})$ denote the message sent from $v_i$ to $g_k$ in iteration $l$ assuming $\boldsymbol{w}$ was received. Similarly, let $m_{ji}^{(l)}(\boldsymbol{w})$ denote the message sent from $s_j$ to $v_i$ in iteration $l$ assuming $\boldsymbol{w}$ was received and let $m_{ki}^{(l)}(\boldsymbol{w})$ denote the message sent from $g_k$ to $v_i$ in iteration $l$ assuming ${\boldsymbol{w}}$ was received. 

    From the variable node symmetry at $l=0$ we have $m_{ij}^{(0)}(\boldsymbol{y}) = x_im_{ij}^{(0)}(\boldsymbol{z})$ and $m_{ik}^{(0)}(\boldsymbol{y}) = x_im_{ik}^{(0)}(\boldsymbol{z})$. Assume now that in iteration $l$ we have $m_{ij}^{(l)}(\boldsymbol{y}) = x_im_{ij}^{(l)}(\boldsymbol{z})$ and $m_{ik}^{(l)}(\boldsymbol{y}) =x_im_{ik}^{(l)}(\boldsymbol{z})$. Since $\boldsymbol{x}$ is a codeword, we have $\prod_{t:\exists e=(v_t,s_j)}x_t=1$. From the SPC node symmetry condition we conclude that
    \begin{equation}
        m_{ji}^{(l+1)}(\boldsymbol{y}) = x_im_{ji}^{(l+1)}(\boldsymbol{z}).
    \end{equation}
    As for GC nodes, the value $\boldsymbol{x}$ of the variable nodes that are connected to a GC node forms a valid codeword of the corresponding subcode $\mathcal{C}$. Thus from the GC node symmetry condition we conclude that in iteration $l+1$ the message sent from GC node $g_k$ to variable node $v_i$ is 
    \begin{equation}
        m_{ki}^{(l+1)}(\boldsymbol{y}) = x_im_{ki}^{(l+1)}(\boldsymbol{z}).
    \end{equation}
    Furthermore, from the variable node symmetry condition, it follows that in iteration $l+1$ the message sent from variable node $v_i$ to SPC node $s_j$ is
    \begin{equation}
        m_{ij}^{(l+1)}(\boldsymbol{y}) = x_im_{ij}^{(l+1)}(\boldsymbol{z})
    \end{equation}
    and the message sent from variable node $v_i$ to GC node $g_k$ is
    \begin{equation}
        m_{ik}^{(l+1)}(\boldsymbol{y}) = x_im_{ik}^{(l+1)}(\boldsymbol{z})
    \end{equation}
    Thus, by induction, all messages to and from variable node $v_i$ when $\boldsymbol{y}$ is received are equal to the product of $x_i$ and the corresponding message when $\boldsymbol{z}$ is received. Hence, both decoders commit exactly the same number of errors (if any), which proves the claim. 
\end{proof}

\quad\par

\begin{lemma}
    In the BMS channel, the APP decoder satisfies the extended symmetry condition.
\end{lemma}
\begin{proof}
    The symmetry conditions of the channel, variable nodes, and SPC nodes are satisfied according to \cite{richardson2001capacity}. Assume that the $\pm1$ sequence $(b_1,\cdots,b_{K})$ is a codeword of the corresponding subcode $\mathcal{C}$. Denote $\boldsymbol{m}_{\thicksim i}$ to be the sequence $(m_1,\cdots,m_{i-1},m_{i+1},\cdots,m_{K})$, $i\in\{1,\cdots,K\}$, and denote $\boldsymbol{b}$ to be the sequence $(b_1,\cdots,b_{K})$, $\boldsymbol{b}_{\thicksim i}$ to be the sequence $(b_1,\cdots,b_{i-1},b_{i+1},\cdots,b_K)$, $i\in\{1,\cdots,K\}$. Then
    \begin{equation}
    \begin{split}
        \Psi_{Gi}^{(l)}(\boldsymbol{m}_{\thicksim i})
        \overset{(a)}{=}&\log\frac{p_{\boldsymbol{M}_{\thicksim i}|X_i}(\boldsymbol{m}_{\thicksim i}|1)}{p_{\boldsymbol{M}_{\thicksim i}|X_i}(\boldsymbol{m}_{\thicksim i}|-1)}\\
        \overset{(b)}{=}&\log\frac{\sum_{\boldsymbol{x'}\in\mathcal{C}:x'_i=1}p_{\boldsymbol{M}_{\thicksim i}|\boldsymbol{X}}(\boldsymbol{m}_{\thicksim i}|\boldsymbol{x'})}{\sum_{\boldsymbol{x'}\in\mathcal{C}:x'_i=-1}p_{\boldsymbol{M}_{\thicksim i}|\boldsymbol{X}}(\boldsymbol{m}_{\thicksim i}|\boldsymbol{x'})}\\
        \overset{(c)}{=}&\log\frac{\sum_{\boldsymbol{x'}\in\mathcal{C}:x'_i=1}p_{\boldsymbol{M}_{\thicksim i}|\boldsymbol{X}}(\boldsymbol{m}_{\thicksim i}\boldsymbol{b}_{\thicksim i}|\boldsymbol{x'}\boldsymbol{b})}{\sum_{\boldsymbol{x'}\in\mathcal{C}:x'_i=-1}p_{\boldsymbol{M}_{\thicksim i}|\boldsymbol{X}}(\boldsymbol{m}_{\thicksim i}\boldsymbol{b}_{\thicksim i}|\boldsymbol{x'}\boldsymbol{b})}\\
        \overset{(d)}{=}&\log\frac{\sum_{\boldsymbol{x'}\in\mathcal{C}:x'_i=b_i}p_{\boldsymbol{M}_{\thicksim i}|\boldsymbol{X}}(\boldsymbol{m}_{\thicksim i}\boldsymbol{b}_{\thicksim i}|\boldsymbol{x'})}{\sum_{\boldsymbol{x'}\in\mathcal{C}:x'_i=-b_i}p_{\boldsymbol{M}_{\thicksim i}|\boldsymbol{X}}(\boldsymbol{m}_{\thicksim i}\boldsymbol{b}_{\thicksim i}|\boldsymbol{x'})}\\
        \overset{(e)}{=}&\log\frac{p_{\boldsymbol{M}_{\thicksim i}|X_i}(\boldsymbol{m}_{\thicksim i}\boldsymbol{b}_{\thicksim i}|b_i)}{p_{\boldsymbol{M}_{\thicksim i}|X_i}(\boldsymbol{m}_{\thicksim i}\boldsymbol{b}_{\thicksim i}|-b_i)}\\
        \overset{(f)}{=}&b_i\Psi_{Gi}^{(l)}(\boldsymbol{m}_{\thicksim i}\boldsymbol{b}_{\thicksim i}), 
    \end{split}
    \end{equation}
    The validity of $(a)$, $(b)$, $(e)$, and $(f)$ is attributed to the definition of the APP decoder. The validity of $(c)$ is attributed to the property of the BMS channel, and the validity of $(d)$ holds due to the linearity of the subcode $\mathcal{C}$. Therefore, the GC node symmetry condition is satisfied.
\end{proof}

\quad\par


In the following sections, we will assume that the all-one codeword is transmitted after modulation, as the symmetry condition is satisfied.

The density function $f$ of the message is considered symmetric if it satisfies the condition $f(m) = e^m f(-m)$. For Gaussian distribution that satisfies the symmetry condition, it can be proven that its variance is equal to twice the mean. Therefore, knowing the mean of the distribution is sufficient to determine the entire Gaussian distribution. In the study by Richardson \cite{richardson2001design}, it was proven that for a given binary-input memoryless output-symmetric channel, the density functions of messages exchanged between variable and check nodes during BP maintain symmetric. This property is established by illustrating that the density of channel messages exhibits symmetry, and this symmetry persists as messages traverse between variable and check nodes. According to \cite[Theorem 4.29]{richardson2008modern}, this symmetry also holds when messages are exchanged at GC nodes using the APP decoder. Consequently, we derive the following lemma.

\quad\par

\begin{lemma}
    For a binary-input memoryless output-symmetric channel, the density functions of messages sent from variable nodes are symmetric for GLDPC codes under the APP decoder. 
    \label{lemma_symmetric_density}
\end{lemma}

\subsection{Monotonicity}
Assume that there is a given class of channels parameterized by $\alpha$, and these channels fulfill the required extended symmetry condition. The parameter $\alpha$ may be real-valued, like the crossover probability $\epsilon$ for the BSC and the standard deviation $\sigma$ for the BI-AWGNC, or may take values in a larger domain.

In \cite{richardson2001capacity}, Richardson and Urbanke showed the monotonicity property for a given LDPC ensemble under 
the BP decoder, with the assumption that the Tanner graph is cycle-free. If for a fixed parameter $\alpha$, the expected fraction of incorrect messages tends to zero with an increasing number of iterations, then this also holds for every parameter $\alpha^{\prime}$ such that $\alpha^{\prime} \leq \alpha$. Therefore, the supremum of all such parameters for which the fraction of incorrect messages tends to zero can be found, which is called the threshold. For any parameter smaller than the threshold, the expected error probability will tend to zero with an increasing number of iterations. This property remains valid for GLDPC codes under the APP decoder.

\quad\par
\begin{lemma}
Let $W$ and $W^{\prime}$ be two given memoryless channels that fulfill the extended channel symmetry condition. Let $W$ and $W^{\prime}$ be represented by their transition probability $p_W(\boldsymbol{y}\mid\boldsymbol{x})$ and $p_{W^{\prime}}(\boldsymbol{y}\mid\boldsymbol{x})$. Assume that $W^{\prime}$ is physically degraded with respect to $W$, which means that $p_{W^{\prime}}(\boldsymbol{y}^{\prime}\mid\boldsymbol{x})=p_Q(\boldsymbol{y}^{\prime}\mid \boldsymbol{y})p_W(\boldsymbol{y}\mid\textbf{x})$ for some auxiliary channel $Q$. For a given GLDPC code and the APP decoder, let $p$ be the expected fraction of incorrect messages passed at the $l$-th decoding iteration under tree-like neighborhoods and transmission over channel $W$, and let $p^{\prime}$ denote the equivalent quantity for transmission over channel $W^{\prime}$. Then $p\leq p^{\prime}$. 
\end{lemma}
\begin{proof}
    We only need to prove that the APP decoder on the GLDPC codes with cycle-free Tanner graphs is equivalent to symbol-by-symbol ML decoding, and the rest of the proof is consistent with \cite[Theorem 1]{richardson2001capacity}.

    The Tanner graph of a GLDPC code is a factor graph. In situations where this factor graph forms a tree structure, the marginalization problem, as computed by the factor graph, can be systematically broken down into progressively smaller tasks, following the tree's structure. These tasks correspond to the computations carried out on variable nodes and constraint nodes within the graph. Since the APP decoder conducts local ML decoding at variable nodes, SPC nodes, and GC nodes, when the decoding is conducted under tree-like neighborhoods, the APP decoding algorithm effectively transforms into a symbol-by-symbol ML decoding algorithm.
\end{proof}

\section{Reduced Gap to Capacity}
In this section, we will use density evolution to analyze the asymptotic performance of GLDPC codes over the BMS channels. However, the complexity of the APP decoding makes the direct computation of message distributions in density evolution challenging and computationally intensive. To address this difficulty, we identify a class of error-correcting block codes in which the propagation of messages from each connected variable node follows a uniform and consistent way when they are used as subcodes on GC nodes. Consequently, the performance analysis and practical decoding of GLDPC codes can be significantly simplified. We refer to such subcodes as message-invariant subcodes. 

We show the results of analysis from density evolution separately on the BEC and BI-AWGN channels. By employing two message-invariant subcodes as examples, we investigated the impact of varying proportions of GC nodes on the threshold. In scenarios over the BI-AWGN channel, we propose a Gaussian mixture approximation algorithm for rapid threshold computation. As an improvement to the Gaussian approximation method for LDPC codes \cite{chung2001analysis}, the Gaussian mixture approximation significantly reduces the loss of accuracy introduced by Gaussian approximation while retaining its low complexity characteristics.

\subsection{Message-invariant Subcodes}
For the APP decoder, the equation (\ref{APP}) for passing information to each connected variable node may not be identical in form, which can introduce extra complexities in the performance analysis and practical decoding operations of GLDPC codes. Subcodes that maintain consistency in the form of passing information from GC nodes to their adjacent variable nodes can simplify the inconvenience caused by APP decoding on the GC nodes. We refer to such subcodes as message-invariant subcodes, and provide their definition formally below.

\quad\par

\begin{definition}[message-invariant subcode]
    A subcode $\mathcal{C}$ is a message-invariant subcode if, when utilizing the APP decoder with $\mathcal{C}$ as the subcode on the GC node, for any variable node $v_i$ connected to $\mathcal{C}$, there exists a permutation $\pi_i$. Applying $\pi_i$ to the variables in the message-passing formula from $\mathcal{C}$ to $v_1$ results in the message-passing formula from $\mathcal{C}$ to $v_i$.
\end{definition}

\quad\par

It can be readily verified that, both the (6,3)-shortened Hamming code and the (7,4)-Hamming code can be classified as message-invariant subcodes. These two codes were all used in \cite{liu2019probabilistic} to show that appropriately proportioned GC nodes can reduce the gap to the channel capacity between the based LDPC codes on the BEC,  We denote the (6,3)-shortened Hamming code as the $\mathcal{C}_1$, and its parity-check matrix $H_1$ is 
\begin{equation}
    H_1 = 
        \begin{blockarray}{ccccccc}
        v_1 &v_2&v_3&v_4&v_5&v_6\\
        \begin{block}{(cccccc)c}
            1&0&0&1&1&0&c_1\\
            0&1&0&1&0&1&c_2\\
            0&0&1&0&1&1&c_3\\
        \end{block}
        \end{blockarray}.
\end{equation}
Denote the (7,4)-shortened Hamming code as the $\mathcal{C}_2$, and its parity-check matrix $H_2$ is 
\begin{equation}
    H_2 = 
        \begin{blockarray}{cccccccc}
        v_1 &v_2&v_3&v_4&v_5&v_6&v_7\\
        \begin{block}{(ccccccc)c}
            0&1&1&1&1&0&0&c_1\\
            1&0&1&1&0&1&0&c_2\\
            1&1&0&1&0&0&1&c_3\\
        \end{block}
        \end{blockarray}.
\end{equation}

In the following sections, we will use $\mathcal{C}_1$ and $\mathcal{C}_2$ as subcodes on the GC nodes for example to conduct density evolution analysis for the GLDPC codes. For more details about performing APP decoding on $\mathcal{C}_1$ and $\mathcal{C}_2$, as well as further analysis and applications of message-invariant subcodes, please refer to Appendix A.

\subsection{Density Evolution on the BEC}

In this subsection, we first analyze the asymptotic performance of $(\mathcal{C}, J, K, t)$ GLDPC ensemble using density evolution on the BEC as an example, where the channel information is erased with probability $\epsilon$. Let $\epsilon_l$ denote the probability of the message sent from the variable nodes in the $l$-th iteration is erased, where $l\geq 0$. According to the definition, the initial V2C message is equivalent to the message received from the channel, which has a probability of $\epsilon$ being erased. Therefore, we have $\epsilon_0 = \epsilon$. Assuming that $\epsilon_l$ is known, we proceed to the erasure probability of C2V messages in the $(l+1)$-th iteration. For SPC nodes, we use $e_{S}^{(l+1)}$ to represent the probability that the message sent by an SPC node in the $(l+1)$-th iteration is erased. The calculation of $e_{S}^{(l+1)}$ aligns with the case on LDPC codes \cite{richardson2008modern}, which is shown that $e_{S}^{(l+1)} = 1-(1-\epsilon_l)^{K-1}$.

As for GC nodes, the probability that the outgoing message is erased needs to be analyzed based on the corresponding subcode. By employing the APP decoder and considering the conditions that the received messages should satisfy to result in erased output messages, we can calculate the probability that the transmitted messages are erased. For message-invariant subcodes, the message-passing formulas to each connected variable node are structurally identical. Therefore, the probabilities that outgoing messages sent to each connected variable node are erased are equal. Let $e_{\mathcal{C}}^{(l+1)}$ denote the probability that the message sent by a GC node with subcode $\mathcal{C}$ in the $(l+1)$-th iteration is erased. In the case where the subcode is $\mathcal{C}_1$, we have 
\begin{equation}
    e_{\mathcal{C}_1}^{(l+1)} = 1-2\epsilon_l^3(1-\epsilon_l)^2-8\epsilon_l^2(1-\epsilon_l)^3-5\epsilon_l(1-\epsilon_l)^4-(1-\epsilon_l)^5,
\end{equation}
while in the case where the subcode is $\mathcal{C}_2$, we have
\begin{equation}
    e_{\mathcal{C}_2}^{(l+1)} = 1 - 4\epsilon_l^3(1-\epsilon_l)^3-12\epsilon_l^2(1-\epsilon_l)^4-6\epsilon_l(1-\epsilon_l)^5-(1-\epsilon_l)^6.
\end{equation}

Given that the proportion of GC nodes is $t$, and the degree of GC nodes matches that of SPC nodes, the probability that an edge uniformly randomly taken from the Tanner graph of the GLDPC ensemble connects to a GC node is $t$, while the probability of it connecting to an SPC node is $1-t$. Averaging over this probability, we obtain the probability that a message received by a variable node in the $(l+1)$-th iteration is erased, expressed as $te_{\mathcal{C}_1}^{(l+1)}+(1-t)e_{S}^{(l+1)}$. Due to the fact that the message sent by a variable node in the $(l+1)$-th iteration is erased only if the channel message and all messages from SPC and GC nodes are erased, we obtain 
\begin{equation}
    \epsilon_{l+1} = \epsilon_0(te_{\mathcal{C}_1}^{(l+1)}+(1-t)e_{S}^{(l+1)})^{J-1}.
\end{equation}
We state this in the following theorem.

\quad\par

\begin{theorem}
    For a BEC channel and $(\mathcal{C}, J, K, t)$ GLDPC ensemble where $\mathcal{C}$ is a message-invariant subcode, let $\epsilon_l$ denote the expected fraction of erased messages passed in the $l$-th iteration under the APP decoder, $l\geq 0$, where $\epsilon_0$ is the probability that the message received from the channel is erased. Then under the independence assumption, the iterative update equation of $\epsilon_l$ is given by
    \begin{equation}
    \epsilon_{l+1} = \epsilon_0(te_{\mathcal{C}}^{(l+1)}+(1-t)e_{S}^{(l+1)})^{J-1}.
    \label{bec_c1}
    \end{equation}
    
\end{theorem}

\quad\par

Using (\ref{bec_c1}), we can determine the threshold $\epsilon^{\ast}$ which is the maximum value of $\epsilon_0$ which can ensure that $\epsilon_l$ will converge to zero with an increasing number of iterations.

\begin{figure}[!t]
\centering
\includegraphics[width=3in]{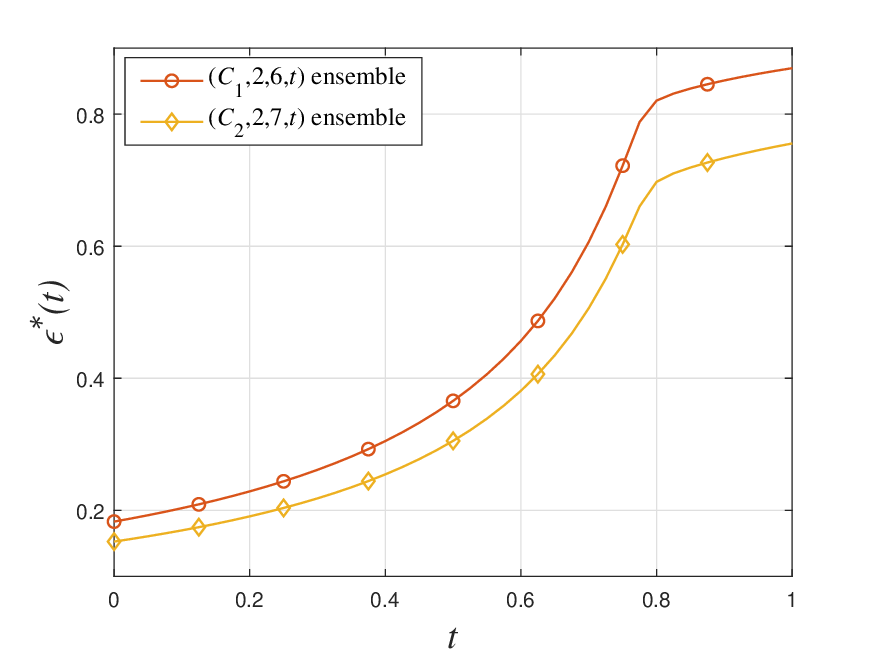}
\caption{The thresholds $\epsilon^{\ast}$ in the BEC as a function of $t$ for both the $(\mathcal{C}_1, 2, 6, t)$ and $(\mathcal{C}_2, 2, 7, t)$ GLDPC ensembles.}
\label{fig_bec_t_ep}
\end{figure}

Similar to \cite{liu2019probabilistic}, we evaluate the asymptotic performance of the GLDPC ensemble as we vary the fraction $t$ of GC nodes in GLDPC. Fig. \ref{fig_bec_t_ep} illustrates the thresholds $\epsilon^{\ast}$ as a function of $t$ for both the $(\mathcal{C}_1,2,6,t)$ and $(\mathcal{C}_2,2,7,t)$ GLDPC ensembles. Note that $\epsilon^{\ast}(t)$ is a continuous, strictly increasing function with respect to $t$. For $t=0$, $\epsilon^{\ast}(0)$ is equal to the threshold of the base LDPC ensemble. Denote the inverse of this function by $t(\epsilon^{\ast})$, which is the minimum fraction of GC nodes in the graph required to achieve the threshold at least $\epsilon^{\ast}$. In a similar way, using (1), we can get the functional relationship between the design rate of GLDPC ensembles and $t$. Therefore, by considering $t$ as a parameter, we can derive the functional relationship between the design rate and threshold, and the functional relationship between the gap to capacity and threshold.

\begin{figure}[!t]
\centering
\subfigure[]{
\includegraphics[width=3in]{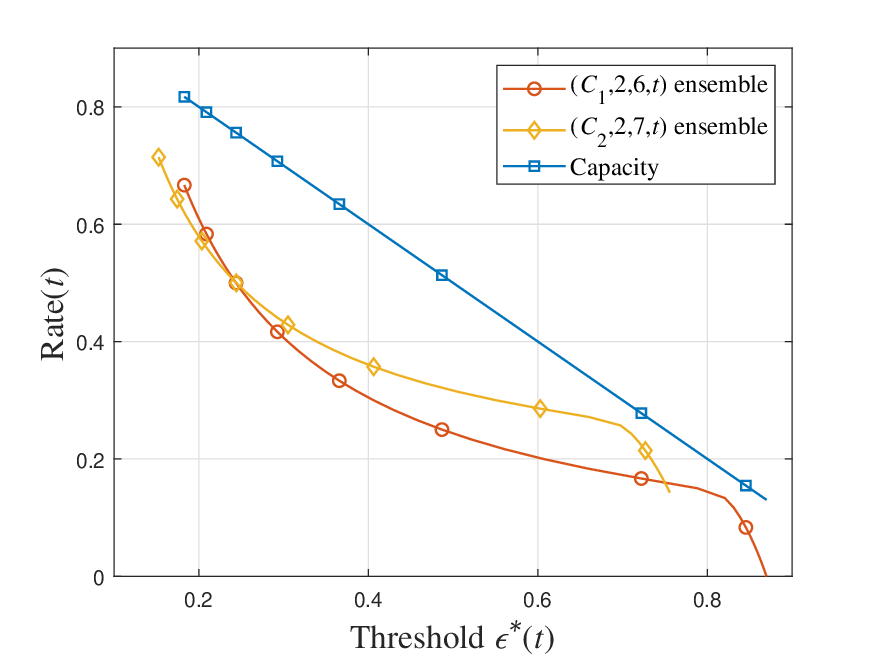} 
}
\subfigure[]{
\includegraphics[width=3in]{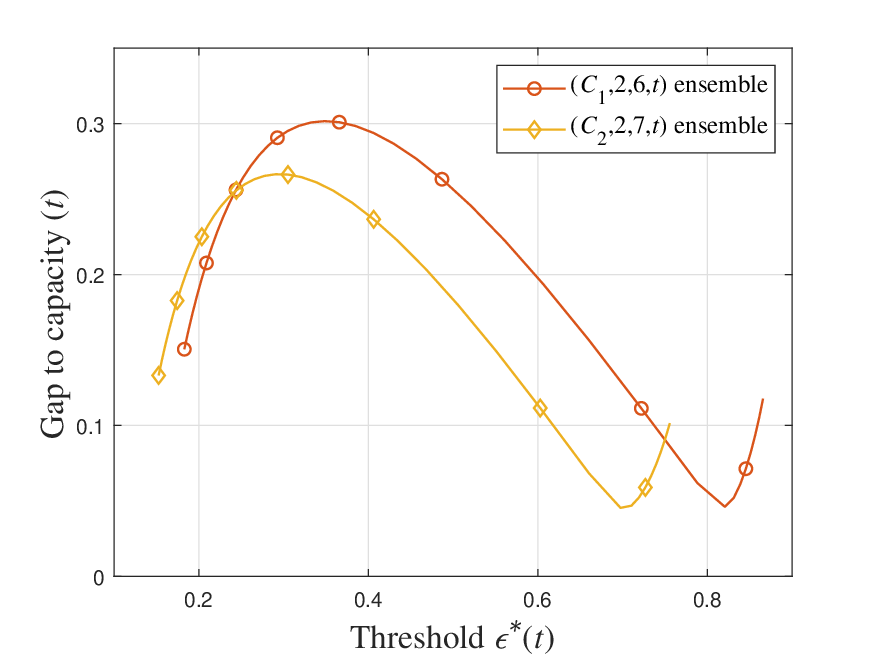} 
}
\DeclareGraphicsExtensions.
\caption{In (a), we show the design rates of $(\mathcal{C}_1,2,6,t)$ ensemble and  $(\mathcal{C}_2,2,7,t)$ ensemble as functions of the threshold under the APP decoder on the BEC, and compare these rates with the channel capacity. In (b), we show the gaps to the channel capacity of these ensembles as functions of the threshold.}
\label{fig_bec_gap}
\end{figure}

In Fig. \ref{fig_bec_gap} (a), we show the design rate of $(\mathcal{C}_1,2,6,t)$ ensemble and  $(\mathcal{C}_2,2,7,t)$ ensemble as a function of the threshold under the APP decoder on the BEC. These code rates are compared against the channel capacity. In Fig. \ref{fig_bec_gap} (b), we show their gap to channel capacity for both ensembles as a function of the threshold. It can be observed that in cases where the proportion of GC nodes $t$ is relatively small, the addition of GC nodes increases the gap to capacity for the GLDPC ensemble. The reason for this is that although the increase in GC nodes can improve performance, it results in a loss of code rate, which overall makes the gap to capacity larger. However, when the proportion of GC nodes becomes significantly higher, the performance improvement brought by GC nodes dominates, leading to a smaller gap to capacity despite the loss in code rate. When the proportion of GC nodes $t$ is greater than 0.73, the gaps to capacity on both the $(\mathcal{C}_1,2,6,t)$ and $(\mathcal{C}_2,2,7,t)$ ensembles will be smaller than their gaps between the base LDPC codes to capacity. Hence, by selecting an appropriate proportion $t$ of GC nodes, it is possible to reduce the gap to capacity on the BEC. This phenomenon is consistent with the observation in \cite{liu2019probabilistic}.

\subsection{Density Evolution on the BI-AWGN Channels}

In BMS channels, the density evolution of GLDPC codes can be analyzed similarly to that on the BEC.

\quad\par

\begin{theorem}
\label{th_awgn}
    For a given BMS channel and $(\mathcal{C}, J, K, t)$ GLDPC ensemble where $\mathcal{C}$ is a message-invariant subcode, let $P_0$ denote the initial message density of log-likelihood ratios, assuming that the all-one codeword was transmitted, and let $P_l$ denote the density of the messages emitted by the variable nodes in the $l$-th iteration under the APP decoder, $l\geq 0$. Then under the independence assumption, the iterative update equation of $P_l$ is given by
    \begin{equation}
    P_{l+1} = P_0\circledast(t\Phi_G^{\mathcal{C}}(P_l)+(1-t)\Phi_S^{K}(P_l))^{\circledast(J-1)},
    \label{eq_de}
    \end{equation}
    where $\Phi_G^{\mathcal{C}}(P_l)$ is the density of the message passed from the GC node with subcode $\mathcal{C}$ at the $(l+1)$-th iteration and $\Phi_S^{K}(P_l)$ is the density of the message passed from the SPC node of degree $K$ at the $(l+1)$-th iteration.  
\end{theorem}

\quad\par

\subsubsection{Gaussian Approximation}
Because the calculation of the messages sent by GC nodes is relatively complex, directly computing $\Phi_G^{\mathcal{C}}(P_l)$ for more intricate channels like the BI-AWGN channels becomes challenging. To simplify the analysis of the density evolution under BI-AWGN channels, we employ the Gaussian approximation for message densities, similar to the approach in \cite{chung2001analysis}. The densities of the messages emitted from variable nodes and SPC nodes can be approximated as Gaussians for the reasons of independent assumption, central limit theorem, and empirical results \cite{chung2001analysis}. For the densities of the messages emitted from GC nodes, experiments show that they can also be well approximated by Gaussian distributions. In Fig. \ref{fig_7}, we show the densities at the GC nodes with subcode $\mathcal{C}_1$ and $\mathcal{C}_2$ respectively, obtained through Monte Carlo methods, where the input messages of the GC nodes follow a Gaussian distribution with a mean of 3 and a variance of 6. It can be seen that the densities at the GC nodes can be well approximated by Gaussian distributions.

\begin{figure}[!t]
\centering
\includegraphics[width=3in]{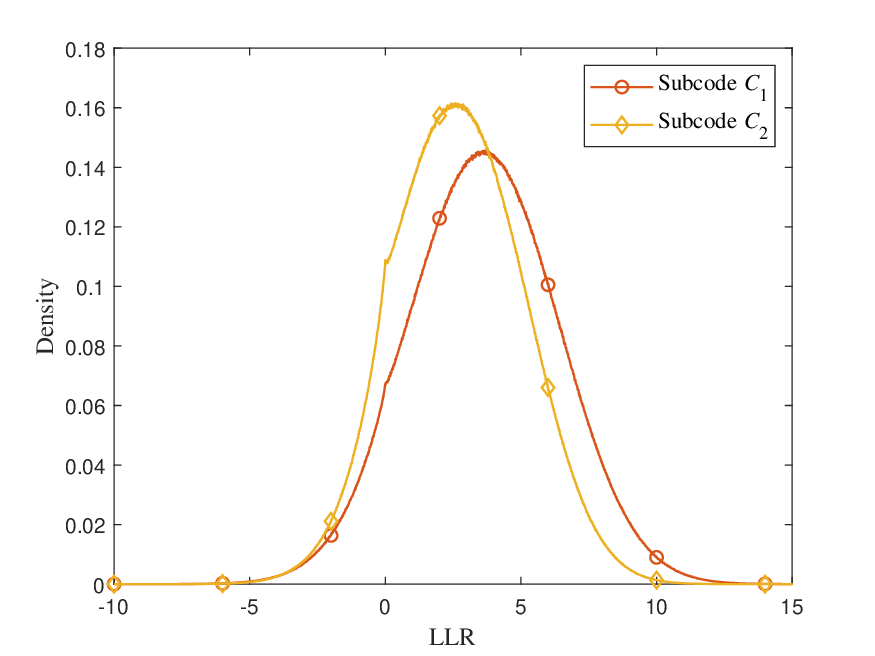}
\caption{The densities at the GC nodes corresponding to subcode $\mathcal{C}_1$ and $\mathcal{C}_2$, respectively, which were obtained through Monte Carlo methods. The input messages of the GC nodes follow a Gaussian distribution with a mean of 3 and a variance of 6. It can be observed that the densities at the GC nodes can be well approximated by Gaussian distributions.}
\label{fig_7}
\end{figure}

Since the message densities sent by variable nodes exhibit symmetry under the APP decoder for GLDPC codes, as established in Lemma \ref{lemma_symmetric_density}, it follows that the distribution of messages emitted by variable nodes with a mean of $m$ has a variance of $2m$. Consequently, it is sufficient to document the mean of the messages to determine the extire distribution, as discussed in \cite{chung2001analysis}.

Denote the mean of density at variable nodes in the $l$-th iteration by $m_{V}^{(l)}$, $l\geq 0$. On the constraint node, we denote the functional relationship between the mean of the output density and the mean of the input density as $\phi_S(m_V)$ and $\phi_G(m_V)$, respectively, for SPC nodes and GC nodes, where the distributions are approximated as Gaussian distributions. Following \cite{chung2001analysis},
\begin{equation}
    \phi_S(m_V) = \phi^{-1}\Bigg(1-\Big[1-\phi(m_V)\Big]^{K-1}\Bigg),
    \label{Gaussian_old_equation}
\end{equation}
where $\phi(m_V)$ is given in (\ref{eq_ga1}). We obtain $\phi_G(m_V)$ using the Monte Carlo method. For computational convenience, approximate forms of $\phi(m_V)$ and $\phi_G(m_V)$ can be utilized. Regarding the approximation for $\phi(m_V)$ and $\phi_G(m_V)$ in the cases of subcode $\mathcal{C}_1$ and $\mathcal{C}_2$, please refer to Appendix B.

By averaging $\phi_S(m_V)$ and $\phi_G(m_V)$ with respect to the parameter $t$, as done in a similar manner to \cite{chung2001analysis} for irregular check nodes, we obtain the following expression:
\begin{equation}
    m_{V}^{(l+1)} = m_{V}^{(0)} + \left(J-1\right)\left(t\phi_G\left(m_V^{(l)}\right)+(1-t)\phi_S\left(m_V^{(l)}\right)\right).
    \label{old_GA}
\end{equation}

However, when $t$ is not equal to 0 or 1, the Gaussian approximation performed using (\ref{old_GA}) sometimes exhibits much inaccuracy. For instance, for the $(\mathcal{C}_1,2,6,0.5)$ ensemble, there is a substantial error of 4.91 dB between the threshold obtained from Gaussian approximation and the threshold obtained from density evolution, where the density of the messages emitted from GC nodes is determined using the Monte Carlo method in each iteration. A similar phenomenon also occurs occasionally in the Gaussian approximation of LDPC codes with irregular check node degrees. For example, for LDPC ensembles with degree distributions given by $\lambda(x) = x^2$ and $\rho(x)=0.9x^2+0.1x^4$, there is an error of 2.38 dB between the threshold obtained from Gaussian approximation and the threshold obtained from density evolution.

\subsubsection{Gaussian mixture Approximation}
The error brought by the Gaussian approximation largely arises from its inability to accurately characterize the distribution of messages sent from variable nodes when different constraint nodes are present. In fact, when we approximate the densities at SPC and GC nodes as Gaussian distributions at iteration $l$, where $l\geq 1$, the average density of messages emitted by the variable nodes should follow a Gaussian mixture distribution. Specifically, the messages received by the variable nodes have a probability of $t$ to follow a Gaussian distribution with a mean of $\phi_G^{(l)}$ and a probability of $1-t$ to follow a Gaussian distribution with a mean of $\phi_S^{(l)}$.

Take $J=2$ as an example. According to Theorem \ref{th_awgn}, for $l\geq 2$,
\begin{equation}
\begin{split}
    P_{l} &= P_0\circledast(t\Phi_G^{\mathcal{C}}(P_{l-1})+(1-t)\Phi_S^{K}(P_{l-1}))\\
    &=t \Big(P_0\circledast \Phi_G^{\mathcal{C}}(P_{l-1})\Big)+
    (1-t) \Big(P_0\circledast \Phi_S^{K}(P_{l-1})\Big).
\end{split}
\end{equation}
By approximating $P_0$ with a Gaussian distribution of mean $m_V^{(0)}$, $\Phi_G^{\mathcal{C}}(P_{l-1})$ with a Gaussian distribution of mean $\phi_G^{(l)}$, $\Phi_S^{K}(P_{l-1})$ with a Gaussian distribution of mean $\phi_S^{(l)}$, $P_{l}$ follows a Gaussian mixture approximation, which has a probability of $t$ to follow a Gaussian distribution with a mean of $\phi_G^{(l)}+m_V^{(0)}$, and a probability of $1-t$ to follow a Gaussian distribution with a mean of $\phi_S^{(l)}+m_V^{(0)}$. In (\ref{old_GA}), this Gaussian mixture distribution is approximated by a Gaussian distribution with a mean of $m_V^{(l)} = t\phi_G^{(l)}+(1-t)\phi_S^{(l)} + m_V^{(0)}$. However, when the values of $\phi_G^{(l)}$ and $\phi_S^{(l)}$ differ significantly, which is common for GLDPC codes due to the difference in the error-correcting capabilities of the subcodes corresponding to SPC nodes and GC nodes, this Gaussian mixture distribution can deviate greatly from a Gaussian distribution, and further leads to inaccurate estimation of the distribution at SPC and GC nodes. 

In Fig. \ref{fig_gaussian_mixture}, for $(\mathcal{C}_1,2,6,0.5)$ and $(\mathcal{C}_2,2,7,0.5)$ GLDPC ensembles, we display the densities of messages sent by variable nodes in the 5-th iteration on a BI-AWGN channel with SNR 4, which are obtained through density evolution as in (\ref{eq_de}), where $\Phi_G^C(P_l)$ is obtained using Monte Carlo methods. It can be seen that the densities deviate significantly from Gaussian distributions, but they can be well approximated by Gaussian mixture distributions.

\begin{figure}[!t]
\centering
\includegraphics[width=3in]{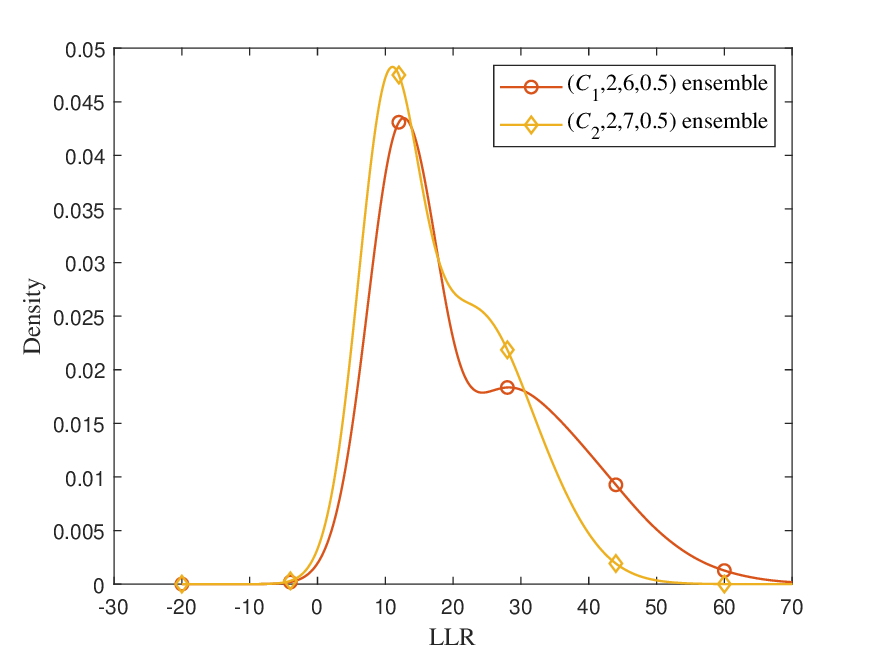}
\caption{The densities of messages sent by variable nodes in the $(\mathcal{C}_1,2,6,0.5)$ ensemble and $(\mathcal{C}_2,2,7,0.5)$ ensemble in the 5-th iteration on the BI-AWGN channel with SNR 4 were obtained through density evolution, as described in equation (\ref{eq_de}). In this process, $\Phi_G^C(P_l)$ was determined using Monte Carlo methods. It is evident that these densities deviate significantly from Gaussian distributions; however, they can be effectively approximated by Gaussian mixture distributions.}
\label{fig_gaussian_mixture}
\end{figure}

By approximating the message distribution at variable nodes as the aforementioned Gaussian mixture distribution, for SPC nodes, we obtain
\begin{equation}
\begin{split}
    \phi_S^{(l+1)} =& \phi^{-1}\Bigg(1-\Big[1-t\Big(\phi[\phi_G^{(l)}+m_V^{(0)}]\Big)\\&-(1-t)\Big(\phi[\phi_S^{(l)}+m_V^{(0)}]\Big)\Big]^{K-1}\Bigg),
\end{split} 
\label{new_S_equation}
\end{equation}
for $l\geq1$.

For a GC node with degree $K$, when it sends messages to one of its neighboring variable nodes, it needs to consider the messages from other $K-1$ neighboring nodes. Each message is drawn from a Gaussian distribution with mean $\phi_G^{(l)}+m_V^{(0)}$ with probability $t$ and from a Gaussian distribution with mean $\phi_S^{(l)}+m_V^{(0)}$ with probability $1-t$. Therefore, with probability $\tbinom{K-1}{\alpha}t^{\alpha}(1-t)^{K-1-\alpha}$, out of the $K-1$ messages, $\alpha$ messages are selected from a Gaussian distribution with mean $\phi_G^{(l)}+m_V^{(0)}$, and $K-1-\alpha$ messages are selected from a Gaussian distribution with mean $\phi_S^{(l)}+m_V^{(0)}$. By averaging over all possible inputs and employing a similar approach as that on SPC nodes, we approximate $\phi_G^{(l+1)}$ as
\begin{equation}
    \phi_G^{(l+1)} = \sum\limits_{\alpha=0}^{K-1}\tbinom{K-1}{\alpha}t^{\alpha}(1-t)^{K-1-\alpha}\phi^{-1}\Bigg(1-F_1^{\alpha}F_2^{K-1-\alpha}\Bigg),
    \label{new_G_equation}
\end{equation}
where 
\begin{equation}
    F_1 = \Bigg(1-\phi\bigg[\phi_G\Big(\phi_G^{(l)}+m_V^{(0)}\Big)\bigg]\Bigg)^{\frac{1}{K-1}}
\end{equation}
and
\begin{equation}
    F_2 = \Bigg(1-\phi\bigg[\phi_G\Big(\phi_S^{(l)}+m_V^{(0)}\Big)\bigg]\Bigg)^{\frac{1}{K-1}}.
\end{equation}

By following the steps outlined above, $\phi_S^{(l+1)}$ and $\phi_G^{(l+1)}$ can be computed iteratively for constraint nodes. Consequently, the error probability of variable node messages in the $(l+1)$-th iteration can be determined.

\begin{table}[!t]
\begin{threeparttable}
\caption{The threshold obtained using different methods for $(\mathcal{C}_1, 2, 6, t)$ GLDPC ensemble and $(\mathcal{C}_2, 2, 7, t)$ GLDPC ensemble}
\centering
\begin{tabular}{|c|c|c|c|c|c|c|}
\hline
$\mathcal{C}$ & $t$ & $\sigma_{Monte}$\tnote{1} & $\sigma_{GA}$\tnote{2}& $\sigma_{GMA}$\tnote{3} & E$_1$[dB]\tnote{4}& E$_2$[dB]\tnote{5}\\
\hline
\hline
$\mathcal{C}_1$ & 0 & 0.5754&0.5857& 0.5857&0.15 &0.15 \\
\hline
$\mathcal{C}_1$ & 0.1 &0.5957&0.6885 &0.6060 &1.40 &0.15\\
\hline
$\mathcal{C}_1$ & 0.3 &0.6539 &0.9461& 0.6641&3.21 &0.13\\
\hline
$\mathcal{C}_1$ & 0.5 & 0.7665&1.3487&0.7732 &4.91 &0.08\\
\hline
$\mathcal{C}_1$ & 0.7 &1.1574&1.8537 &1.1382 &4.10 &0.15\\
\hline
$\mathcal{C}_1$ & 0.9 &2.1478 &2.2346& 2.1605&0.34 &0.05\\
\hline
$\mathcal{C}_1$ & 1 &2.3550 &2.4046 &2.4060&0.18 &0.18\\
\hline
\hline
$\mathcal{C}_2$ & 0 &0.5464 & 0.5556& 0.5556&0.17 &0.14\\
\hline
$\mathcal{C}_2$ & 0.1 &0.5636&0.6377 & 0.5729& 1.05&0.14\\
\hline
$\mathcal{C}_2$ & 0.3 & 0.6116& 0.8444&0.6209 & 2.80&0.13\\
\hline
$\mathcal{C}_2$ & 0.5 &0.7006 & 1.1076& 0.7047& 3.98&0.05\\
\hline
$\mathcal{C}_2$ & 0.7 &0.9627 &1.3142&0.9151 & 2.70&0.44\\
\hline
$\mathcal{C}_2$ & 0.9 & 1.5101 &1.4592& 1.4959 & 0.30&0.08\\
\hline
$\mathcal{C}_2$ & 1 & 1.6106&1.5497 &1.6448&0.18 &0.18\\
\hline
\end{tabular}
\begin{tablenotes}
\footnotesize
	\item[1] $\sigma_{Monte}$ is the threshold computed through density evolution, where the density of the messages emitted from GC nodes is obtained by applying the Monte Carlo method in each iteration.
	\item[2] $\sigma_{GA}$ is the threshold computed through Gaussian approximation as in (\ref{old_GA}).
 \item[3] $\sigma_{GMA}$ is the threshold computed through Gaussian mixture approximation as in (\ref{new_S_equation}) and (\ref{new_G_equation}).
 \item[4]  E$_1$[dB] represents the difference, in dB, between the threshold computed using Gaussian approximation and the threshold computed using density evolution with Monte Carlo.
  \item[5]  E$_2$[dB] represents the difference, in dB, between the threshold computed using Gaussian mixture approximation and the threshold computed using density evolution with Monte Carlo.
			\end{tablenotes}
   \end{threeparttable}
   \label{table_gaussian_mixture}
\end{table}

We compute the thresholds for both the $(\mathcal{C}_1, 2, 6, t)$ and $(\mathcal{C}_2, 2, 7, t)$ GLDPC ensembles on the BI-AWGN channel using three different methods: density evolution as described in (\ref{eq_de}), Gaussian approximation as described in (\ref{old_GA}), and Gaussian mixture approximation following (\ref{new_S_equation}) and (\ref{new_G_equation}). For density evolution, we determined the density of messages from GC nodes by employing the Monte Carlo method in each iteration. The results are presented in TABLE \uppercase\expandafter{\romannumeral1}. It can be noted that for both $\mathcal{C}_1$ and $\mathcal{C}_2$, the Gaussian mixture approximation method can considerably reduce the errors compared to Gaussian approximation across various values of $t$. The error in the Gaussian mixture approximation in Table \uppercase\expandafter{\romannumeral1} can be further reduced by obtaining a more finely accurate estimation for $\phi_G(m_V)$. It is worth noting that the aforementioned Gaussian mixture approximation method can be similarly applied to LDPC codes with irregular check node degrees. In TABLE \uppercase\expandafter{\romannumeral2}, we apply the Gaussian mixture approximation method to LDPC codes ensembles, where the variable node degree is set to 3, and the check nodes have degrees of 3 and 5. The comparison reveals a significant improvement in threshold estimation accuracy over Gaussian approximation. In this specific example, the maximum error between the threshold obtained by Gaussian approximation and density evolution was 2.38 dB, whereas the maximum error between the threshold obtained by Gaussian mixture approximation and density evolution was 0.15 dB.

\begin{table}[!t]
\begin{threeparttable}
\caption{The threshold obtained using different methods for LDPC codes.}
\centering
\begin{tabular}{|c||c|c||c|c|c|c|c|}
\hline
$\lambda_3$\tnote{1} & $\rho_3$\tnote{2}& $\rho_5$& $\sigma_{DE}$\tnote{3}& $\sigma_{GA}$& $\sigma_{GMA}$& E$_1$[dB]\tnote{4}& E$_2$[dB]\tnote{5}\\
\hline
1 & 0 & 1& 1.0059& 0.9983& 0.9983 & 0.06&0.06\\
\hline
1 & 0.1 & 0.9& 1.0645& 1.0480& 1.0509& 0.13 & 0.10\\
\hline
1 & 0.3 & 0.7& 1.2051& 1.1480&1.1840& 0.42 &0.15\\
\hline
1 & 0.5 & 0.5& 1.3926& 1.2570& 1.3691&0.89 &0.15\\
\hline
1 & 0.7 & 0.3& 1.6504& 1.3710&1.6216& 1.61 &0.15\\
\hline
1 & 0.9 & 0.1& 1.9551& 1.4870&1.9324& 2.38 &0.10\\
\hline
\end{tabular}
\begin{tablenotes}
\footnotesize
	\item[1] $\lambda_3$ represents the proportion of edges connected to degree-3 variable nodes in the graph.
	\item[2] $\rho_3$ represents the proportion of edges connected to degree-3 check nodes in the graph.
 \item[3] $\sigma_{DE}$ is the threshold computed through density evolution \cite{richardson2001capacity}.
 \item[4]  E$_1$[dB] represents the difference, in dB, between the threshold computed using Gaussian approximation and the threshold computed using density evolution.
 \item[5]  E$_2$[dB] represents the difference, in dB, between the threshold computed using Gaussian 
mixture approximation and the threshold computed using density evolution.
			\end{tablenotes}
   \end{threeparttable}
\end{table}

Through the aforementioned Gaussian mixture approximation method, we can calculate the threshold of channel parameters for a given GLDPC ensemble while varying the parameter $t$. Figure \ref{fig_9} illustrates these thresholds, denoted as SNR$^\ast$, as a function of $t$ for both the $(\mathcal{C}_1,2,6,t)$ and $(\mathcal{C}_2,2,7,t)$ GLDPC ensembles in BI-AWGN channels. We see that SNR$^{\ast}(t)$ is a continuous, strictly decreasing function of $t$. Denote the inverse of this function by $t($SNR$^{\ast})$, which is the maximum fraction of GC nodes in the graph required to achieve an ensemble threshold at most SNR$^{\ast}$. By employing (1), we can establish a functional connection between the design rates of GLDPC ensembles and the parameter $t$. Consequently, by treating $t$ as a parameter in our analysis, we can derive functional relationships between the design rate and threshold, as well as the functional relationships between the gap to capacity and threshold.

In Fig. \ref{fig_awgn_gap} (a), we plot the design rate of the $(\mathcal{C}_1,2,6,t)$ and $(\mathcal{C}_2,2,7,t)$ ensembles as functions of the threshold under the APP decoder on the BI-AWGN channel, and compare their design rates with the channel capacity. In Fig. \ref{fig_awgn_gap} (b), we display their gaps to the channel capacity. Since SNR$^{\ast}(t)$ is a monotonically decreasing function with $t$ as shown in Fig. \ref{fig_9}, in Fig. \ref{fig_awgn_gap}, as the abscissa increases from left to right, $t$ continuously decreases. The rightmost points on the curves in Fig. \ref{fig_awgn_gap} correspond to scenarios where $t$ equals 0. Consistent with the observations over the BEC, when the proportion of $t$ is relatively small, there is an increase in the gap to the channel capacity due to the loss in code rate compared to the base LDPC code. However, when the proportion of $t$ is larger and appropriately chosen, despite the loss in code rate, the decoding improvements outweigh the rate loss, leading to a smaller gap to the channel capacity in comparison to the base LDPC code. For the $(\mathcal{C}_1,2,6,t)$ ensemble and the $(\mathcal{C}_2,2,7,t)$ ensemble, their minimum gaps to capacity are achieved at $t=0.8$ and $t=0.85$, respectively.

\begin{figure}[!t]
\centering
\includegraphics[width=3in]{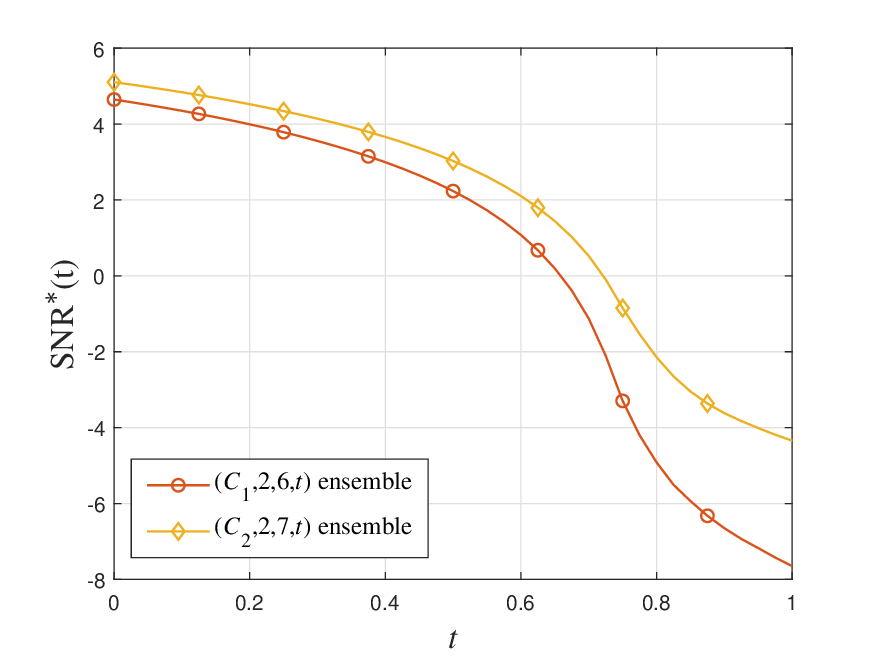}
\caption{The thresholds in the BI-AWGN channel as a function of $t$ for $(\mathcal{C}_1, 2, 6, t)$ and $(\mathcal{C}_2, 2, 7, t)$ GLDPC ensembles.}
\label{fig_9}
\end{figure}

\begin{figure}[!t]
\centering
\subfigure[]{
\includegraphics[width=3in]{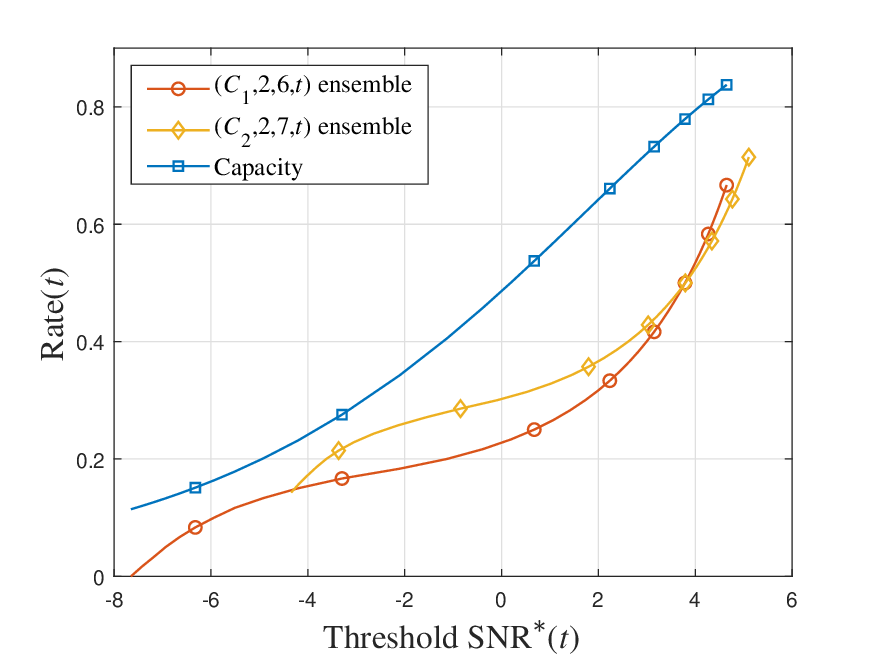} 
}
\subfigure[]{
\includegraphics[width=3in]{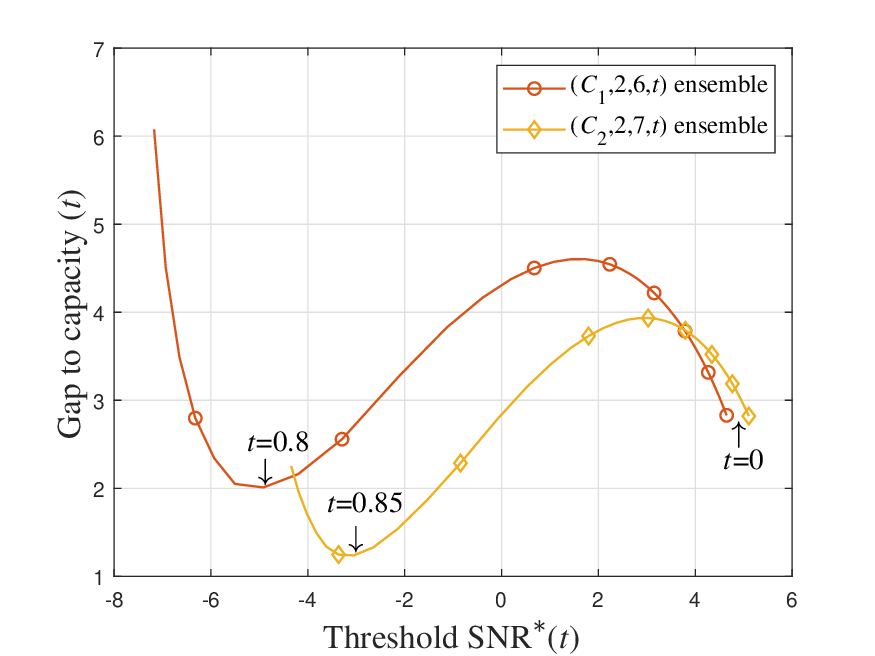} 
}
\DeclareGraphicsExtensions.
\caption{In (a), we plot the design rate of $(\mathcal{C}_1,2,6,t)$ ensemble and  $(\mathcal{C}_2,2,7,t)$ ensemble as a function of the threshold under the APP decoder on the BI-AWGN channel, and compare these rates with the channel capacity. In (b), we show their gaps to the channel capacity.}
\label{fig_awgn_gap}
\end{figure}

\subsection{Simulation Results}
In this section, we construct GLDPC codes with appropriate proportion of GC nodes and compare their performance through simulation experiments with LDPC codes which are at the same design rate as GLDPC codes. 

In Fig. \ref{fig_simulation_1}, we uniformly randomly select a GLDPC code with a code length of 3000 and subcode $\mathcal{C}_1$. The Tanner graph of the selected GLDPC code is free of cycles of length 4 and parallel edges, which has variable nodes of degree 2 and check nodes of degree 6, with GC nodes ratio $t=0.8$. The maximum number of iterations is set to be 20. For comparison purposes, an LDPC code is obtained by treating the parity-check matrix of this GLDPC code as the parity-check matrix of an LDPC code. We randomly permute the edges on this matrix to remove any cycles of length 4, enabling us to perform BP decoding on this LDPC code which has the same code rate as the GLDPC code. It can be observed that compared to this LDPC code with the same design rate, the GLDPC code exhibits a much lower block error rate (BLER), with an improvement of over 1dB.

Similarly, in Fig. \ref{fig_simulation_2}, we uniformly randomly selected a GLDPC code of length 3500 with subcode $\mathcal{C}_2$, whose Tanner graph is free of cycles of girth 4 and parallel edges. The GLDPC code has variable nodes of degree 2 and check nodes of degree 7, with GC nodes ratio $t=0.85$. The maximum number of iterations is set to be 50. The BLER of the LDPC obtained by eliminating the cycles of length 4 in the parity-check matrix of this GLDPC code is shown for comparison. It can be observed that GLDPC codes exhibit a gain of approximately 0.4dB compared to LDPC codes with the same design rate.

\begin{figure}[!t]
\centering
\includegraphics[width=3in]{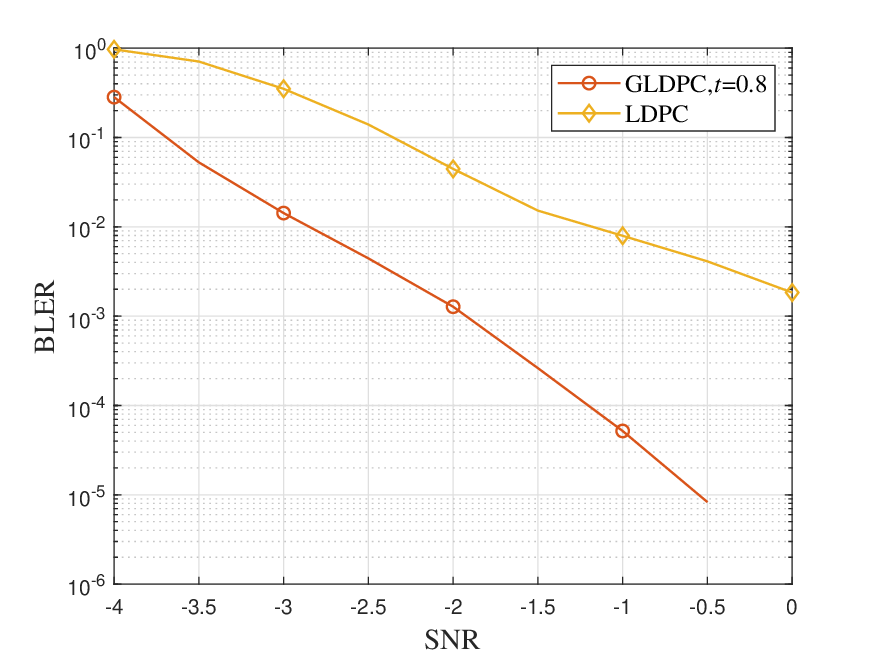}
\caption{The BLER of a randomly chosen (2,6)-GLDPC code which is free of cycles of length 4 in its Tanner graph. The code has length 3000 and subcode $\mathcal{C}_1$, where $t$ is set to be 0.8. The maximum number of iterations is set to be 20. Together is the BLER of the LDPC code which is obtained by treating the parity-check matrix of this GLDPC code as the parity-check matrix of an LDPC code. Cycles of length 4 in the Tanner graph of this LDPC code are eliminated by randomly permuting the edges in the graph. It is worth noting that the GLDPC code and LDPC code used for comparison have the same design rate.}
\label{fig_simulation_1}
\end{figure}

\begin{figure}[!t]
\centering
\includegraphics[width=3in]{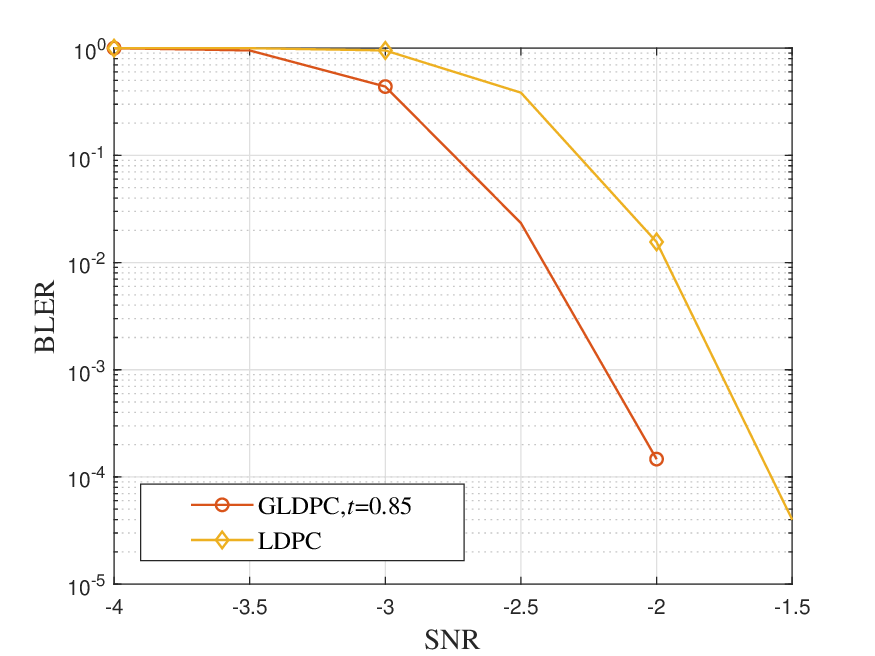}
\caption{The BLER of a randomly chosen (2,7)-GLDPC code which is free of cycles of length 4 in its Tanner graph. The code has length 3500 and subcode $\mathcal{C}_1$, where $t$ is set to be 0.85. The maximum number of iterations is set to be 50. Together is the BLER of the LDPC code which is obtained by treating the parity check matrix of this GLDPC code as the parity-check matrix of an LDPC code. Cycles of length 4 in the Tanner graph of this LDPC code are eliminated by randomly permuting the edges in the graph. It is worth noting that the GLDPC code and LDPC code used for comparison have the same design rate.}
\label{fig_simulation_2}
\end{figure}

\section{Conclusion}
In this study, we analyze the performance of GLDPC codes under the APP decoder by extending the methods of density evolution from LDPC codes to GLDPC codes. Similar to density evolution on LDPC codes, the concentration condition, symmetry condition, and monotonicity condition in GLDPC codes under the APP decoder can be established to provide theoretical guarantees for density evolution algorithms.

In particular, we define a class of message-invariant subcodes, which can significantly reduce the complexity of theoretical analysis and practical decoding of GLDPC codes under the APP decoder. Using two message-invariant subcodes as examples, we provide density evolution analysis on GLPDC codes for the BMS channels. For both the BEC and BI-AWGN channels, we illustrate that by appropriately selecting the fraction of GC nodes, GLDPC codes can achieve a reduced gap to capacity compared to the base LDPC code for both channels.

Among them, on the BI-AWGN channel, we propose a Gaussian mixture approximation method as a fast approximation algorithm for density evolution. Compared to the Gaussian approximation method, the Gaussian mixture approximation can significantly reduce the errors introduced by approximation while still having a low complexity similar to Gaussian approximation. Furthermore, this approximation method can also be extended to LDPC codes. 

Looking forward, further research is needed to delve into the utilization of density evolution for selecting appropriate subcodes and constructing GLDPC codes that approach channel capacity while remaining within acceptable decoding complexity. Additionally, more investigation is needed to explore the application of density evolution in analyzing the error floor of GLDPC codes and exploring related aspects.

\begin{appendices}
\section{Message-invariant subcodes}
To show that $\mathcal{C}_1$ is a message-invariant subcode, examine the Tanner graph of $\mathcal{C}_1$ provided in Fig. \ref{fig_2}. It can be observed that $v_1$, $v_2$, and $v_3$ exhibit a symmetrical relationship in the Tanner graph. 
That is, the formulas for passing messages to $v_1$, $v_2$ and $v_3$ differ only in the order of input variables, as in (\ref{subcode1_l1}) and (\ref{subcode1_l2}) as an example, where $l_i$ is the LLR value of $v_i$. 
\begin{equation}
l_1 = \log \frac{e^{l_2+l_3+l_4+l_5+l_6}+e^{l_2+l_3}+e^{l_4+l_5}+e^{l_6}}{e^{l_2+l_4+l_6}+e^{l_3+l_5+l_6}+e^{l_2+l_5}+e^{l_3+l_4}}.\label{subcode1_l1}
\end{equation}

\begin{equation}
l_2 = \log \frac{e^{l_1+l_3+l_4+l_5+l_6}+e^{l_1+l_3}+e^{l_4+l_6}+e^{l_5}}{e^{l_1+l_4+l_5}+e^{l_3+l_5+l_6}+e^{l_1+l_6}+e^{l_3+l_4}}.
\label{subcode1_l2}
\end{equation}

Therefore, the messages sent from GC node to its connected variable nodes $v_1$, $v_2$ and $v_3$ should have the same form. This means that when transmitting messages to $v_2$ and $v_3$, the formulas of passing messages to $v_1$ can be used by changing the order of the input messages correspondingly. For example, when transmitting messages to $v_2$, we can use the formulas of passing messages to $v_1$ by mapping the input $l_1, l_3, l_4, l_5, l_6$ to $l_2, l_3, l_4, l_6, l_5$ in (\ref{subcode1_l1}). Similarly, when GC nodes transmit messages to $v_4$, $v_5$, and $v_6$, they should also have the same form. We provide the formulas for the messages transmitted from GC node to $v_4$ as below.

\begin{figure}[!t]
\centering
\includegraphics[width=2.5in]{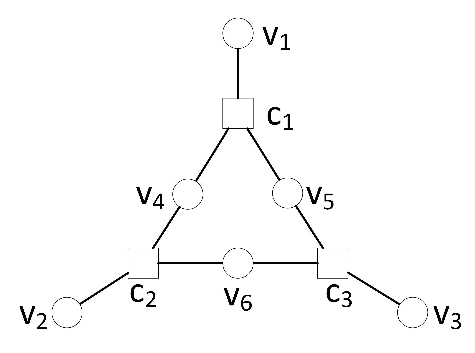}
\caption{The Tanner graph of $\mathcal{C}_1$.}
\label{fig_2}
\end{figure}

\begin{equation}
l_4 = \log \frac{e^{l_1+l_2+l_3+l_5+l_6}+e^{l_1+l_5}+e^{l_2+l_6}+e^{l_3}}{e^{l_1+l_2+l_3}+e^{l_3+l_5+l_6}+e^{l_2+l_5}+e^{l_1+l_6}}.
\end{equation}
It can be observed that the formulas for passing messages to $v_4$ also share the same form as the formulas for passing messages to $v_1$. In fact, we can find a permutation $\pi_4$ such that applying $\pi_4$ to the variables in the message-passing formula for $v_1$ yields the formula for message-passing to $v_4$, where $\pi_4$ is defined as $\pi_4(l_1) = l_4$, $\pi_4(l_2) = l_2$, $\pi_4(l_3) = l_6$, $\pi_4(l_4) = l_1$, $\pi_4(l_5) = l_5$, and $\pi_4(l_6) = l_3$.

Therefore, GC nodes with $\mathcal{C}_1$ as the subcode transmit messages to their neighboring variable nodes in a consistent manner in terms of form, indicating that $\mathcal{C}_1$ is a message-invariant subcode. 

Hence, once $\mathcal{C}_1$ has received all the messages from its connected variable nodes, in order to convey a message to $v_i$, we merely need to apply the appropriate permutation to the input information in the reverse order of $\pi_i$. Subsequently, we can employ the formula designed for transmitting messages to $v_1$ for the message-passing process. This characteristic can simplify the decoding process and analysis. We record the required permutations $\pi_i$ corresponding to $v_i$ for $\mathcal{C}_1$ in TABLE \uppercase\expandafter{\romannumeral3}. 

For the subcode $\mathcal{C}_2$, through similar analysis, we can derive the formula for sending messages to each of its adjacent variables based solely on the formula for sending messages to $v_1$ and the corresponding permutations. We provide the message passing equations and corresponding permutations needed to convert the messages sent to each variable node into messages passed to $v_1$, as shown in the following equation and Table \uppercase\expandafter{\romannumeral4}. 
\begin{equation}
\begin{split}
        l_1 &= \log (e^{l_2+l_3+l_4+l_5+l_6+l_7}+e^{l_2+l_4+l_7}+e^{l_2+l_5+l_6}+e^{l_3+l_4+l_6}\\
        &\quad\quad+e^{l_3+l_5+l_7}+e^{l_2+l_3}+e^{l_4+l_5}+e^{l_6+l_7})\\
        &-\log(e^{l_2+l_3+l_4+l_5}+e^{l_2+l_3+l_6+l_7}+e^{l_4+l_5+l_6+l_7}+e^{l_2+l_4+l_6}\\
        &\quad\quad+e^{l_2+l_5+l_7}+e^{l_3+l_4+l_7}+e^{l_3+l_5+l_6}+1),
\end{split}
\label{subcode2_l1}
\end{equation}

\begin{table}[!t]
\caption{The permutations for $\mathcal{C}_1$\label{tab:table1}}
\centering
\begin{tabular}{|c||c|c|c|c|c|c|}
\hline
$\pi_i$ & $\pi_i(l_1)$& $\pi_i(l_2)$& $\pi_i(l_3)$& $\pi_i(l_4)$& $\pi_i(l_5)$& $\pi_i(l_6)$ \\
\hline
$\pi_1$ & $l_1$ & $l_2$& $l_3$& $l_4$& $l_5$& $l_6$ \\
\hline
$\pi_2$ & $l_2$ & $l_3$& $l_1$& $l_6$& $l_4$& $l_5$ \\
\hline
$\pi_3$ & $l_3$ & $l_1$& $l_2$& $l_5$& $l_6$& $l_4$ \\
\hline
$\pi_4$ & $l_4$ & $l_2$& $l_6$& $l_1$& $l_5$& $l_3$ \\
\hline
$\pi_5$ & $l_5$ & $l_1$& $l_4$& $l_3$& $l_6$& $l_2$ \\
\hline
$\pi_6$ & $l_6$ & $l_3$& $l_5$& $l_2$& $l_4$& $l_1$ \\
\hline
\end{tabular}
\label{C1}
\end{table}

\begin{table}[!t]
\caption{The permutations for $\mathcal{C}_2$\label{tab:table1}}
\centering
\begin{tabular}{|c||c|c|c|c|c|c|c|}
\hline
$\pi_i$ & $\pi_i(l_1)$& $\pi_i(l_2)$& $\pi_i(l_3)$& $\pi_i(l_4)$& $\pi_i(l_5)$& $\pi_i(l_6)$ & $\pi_i(l_7)$\\
\hline
$\pi_1$ & $l_1$ & $l_2$& $l_3$& $l_4$& $l_5$& $l_6$ & $l_7$\\
\hline
$\pi_2$ & $l_2$ & $l_3$& $l_1$& $l_4$& $l_6$& $l_7$ & $l_5$\\
\hline
$\pi_3$ & $l_3$ & $l_1$& $l_2$& $l_4$& $l_7$& $l_5$ & $l_6$\\
\hline
$\pi_4$ & $l_4$ & $l_7$& $l_3$& $l_1$& $l_5$& $l_6$ & $l_2$\\
\hline
$\pi_5$ & $l_5$ & $l_2$& $l_7$& $l_4$& $l_1$& $l_6$ & $l_3$\\
\hline
$\pi_6$ & $l_6$ & $l_2$& $l_4$& $l_3$& $l_5$& $l_1$ & $l_7$\\
\hline
$\pi_7$ & $l_7$ & $l_4$& $l_3$& $l_2$& $l_5$& $l_6$ & $l_1$\\
\hline
\end{tabular}
\label{C2}
\end{table}

Concerning the characteristics of a code that qualifies as a message-invariant subcode, we introduce the following lemma.

\quad\par

\begin{lemma}
The automorphism group Aut($\mathcal{C}$) of code $\mathcal{C}$ is the largest group of $n \times n$ permutation matrices that preserve the codewords of $\mathcal{C}$. A code is termed "transitive" if its automorphism group acts transitively on its codewords. Then, a transitive code $\mathcal{C}$ is a message-invariant subcode.
\end{lemma}
\begin{proof}
    From the definition, for any variable node $v_i$ linked to $\mathcal{C}$, we can identify an automorphism $\pi_i$ of $\mathcal{C}$ such that $\pi_i(v_1) = v_i$. Therefore, $\pi_i$ is the permutation we are looking for with respect to $v_i$.
\end{proof}

\quad\par

There are many established examples of transitive codes, such as Hamming codes, extended Hamming codes, Reed–Muller codes, extended BCH codes, and extended Preparata codes as documented in \cite{mogilnykh2020coordinate}. It can be shown that both $\mathcal{C}_1$ and $\mathcal{C}_2$ are indeed transitive codes.


\section{Approximation Formulas in Gaussian Approximation}
Following the method in \cite{chung2001analysis}, for SPC nodes, $\phi(m_V)$ is approximated by 
\begin{equation}
    \phi(m_V) \approx\left\{
\begin{aligned}
&   e^{-0.4527m_V^{0.86}+0.0218},\;\quad\quad if\; 0< m_V < 10,\\
&   \sqrt{\frac{\pi}{m_V}} e^{-\frac{m_V}{4}}(1-\frac{10}{7m_V}), \; else.
\end{aligned}
\right.
\end{equation}

For GC nodes, $\phi_G(m_V)$ is obtained using the Monte Carlo method. For subcode $\mathcal{C}_1$, 
\begin{equation}
    \phi_G(m_V) \approx\left\{
\begin{aligned}
&   -0.22m_V^3+0.86m_V^2+0.022m_V,\; if\; 0< m_V \leq1,\\
&   0.20m_V^2+0.75m_V-0.28,\;\;\;\;\;\;\;\;\;\;\; if\; 1< m_V \leq2,\\
&   0.042m_V^2+1.4m_V-1,\quad\quad\quad\quad\; if\; 2< m_V \leq5,\\
&   1.9m_V-3, \; \quad\quad\quad\quad\quad\quad \quad\quad\;\; else.
\end{aligned}
\right.
\end{equation}
For subcode $\mathcal{C}_2$,
\begin{equation}
    \phi_G(m_V) \approx\left\{
\begin{aligned}
&   -0.183m_V^4+0.375m_V^3+0.149m_V^2-0.015m_V,\;\\
&\;\;\;\;\;\;\;\;\;\;\;\;\;\;\;\;\;\;\;\;\;\;\;\;\;\;\;\;\;\;\;\;\;\;\;\;\;\;\;\;\;\;\;\;\;\;\;\;\;\;\;if\; 0< m_V \leq1,\\
&   -0.013m_V^4+0.013m_V^3+0.3634m_V^2-0.038,\;\\
&\;\;\;\;\;\;\;\;\;\;\;\;\;\;\;\;\;\;\;\;\;\;\;\;\;\;\;\;\;\;\;\;\;\;\;\;\;\;\;\;\;\;\;\;\;\;\;\;\;\;\;if\; 1< m_V \leq2,\\
&   0.0024m_V^4-0.051m_V^3+0.421m_V^2+0.064m_V-0.11,\;\\
&\;\;\;\;\;\;\;\;\;\;\;\;\;\;\;\;\;\;\;\;\;\;\;\;\;\;\;\;\;\;\;\;\;\;\;\;\;\;\;\;\;\;\;\;\;\;\;\;\;\;\;if\; 2< m_V \leq5,\\
&   0.0025m_V^2+1.71m_V-2.889, \quad\quad\;\; else.
\end{aligned}
\right.
\end{equation}

\end{appendices}

\bibliographystyle{unsrt}
\bibliography{reference}

\end{document}